\documentclass[12pt]{article}

\usepackage{xr}
\usepackage[utf8]{inputenc}
\usepackage{amsmath}
\usepackage{amssymb}
\usepackage{physics}
\usepackage{mathtools}
\usepackage{comment}
\usepackage{amsthm}
\usepackage{bbm}
\usepackage{xcolor}
\usepackage{tikz-cd} 
\usepackage{adjustbox}
\usepackage[parfill]{parskip}
\usepackage{soul}
\usepackage[margin=1.3755in]{geometry}
\newtheorem{theorem}{Theorem}[section]
\newtheorem{corollary}{Corollary}[theorem]

\newcommand{\R}{\mathbb{R}}
\newcommand{\set}[1]{\{ #1 \}}
\newcommand{\dprod}[2]{\langle #1, #2\rangle}
\begin{document}
\title{Graph Embeddings via Tensor Products and Approximately Orthonormal Codes}
\author{Frank Qiu \thanks{Statistics Department; University of California, Berkeley}}
\date{}

\maketitle
\begin{abstract}%   <- trailing '%' for backward compatibility of .sty file
We propose a dynamic graph representation method, showcasing its rich representational capacity and establishing some of its theoretical properties. Our representation falls under the bind-and-sum approach in hyperdimensional computing (HDC), and we show that the tensor product is the most general binding operation that respects the superposition principle employed in HDC. We also establish some precise results characterizing the behavior of our method, including a memory vs. size analysis of how our representation's size must scale with the number of edges in order to retain accurate graph operations. True to its HDC roots, we also compare our graph representation to another typical HDC representation, the Hadamard-Rademacher scheme, showing that these two graph representations have the same memory-capacity scaling. We establish a link to adjacency matrices, showing that our method is a pseudo-orthogonal generalization of adjacency matrices. In light of this, we briefly discuss its applications toward a dynamic compressed representation of large sparse graphs.
\end{abstract}

\newpage
\section{Introduction}
Graphs are a central data structure that pop up in a variety of applications, and informative representations of graphs are an popular area of study \cite{Kang2022RelHDAG} \cite{Nickel2016HolographicEO} \cite{Nunes2022GraphHDEG}. Our proposed graph representation can be motivated from two different angles. 

Our graph representation can be viewed as a pseudo-orthogonal generalization of adjacency matrices.The adjacency matrix $A$ of a graph $G$ can be expressed as the sum of the coordinate tensors:
\[
A = \sum_{ij: A_{ij} = 1} e_i \otimes e_j = \sum_{ij: A_{ij} = 1} e_i e_j^T
\]
where $e_i$ are the standard coordinate vectors. Consider another orthonormal basis $B = \set{b_1,\cdots,b_d}$ with the associated change of basis matrix $P$. Changing bases to $B$, the coordinate form $A_B$ of the original adjacency matrix $A$ in the new basis satisfies the following relation:
\[
A_B = \sum_{ij: A_{ij}=1} b_i b_j^T = P A P^T
\]
Matrices that can be expressed as sums of simple orthonormal tensors are equivalent to adjacency matrices, up to an orthogonal change of basis. \textbf{Our insight is to replace the orthonormal vectors with pseudo-orthonormal vectors.} While one can pack at most $d$ orthornormal vectors into $\R^d$, we show that we can pack $O(e^{C d \epsilon^2 })$ $\epsilon$-orthogonal unit vectors in $\R^d$ through the use of spherical codes, which are random vectors sampled uniformly from $\mathbb{S}^{d-1}$. This has implications for the drastic compression of large sparse adjacency matrices, and our method can be viewed as a compressed analogue of the adjacency matrix.

Our method was originally inspired from hyperdimensional computing, which take the bind-and-sum approach to representing graphs. Bind-and-sum methods seek to represent a graph's edgeset through a three-step process:
\begin{enumerate}
    \item Encode each vertex as a vector
    \item Encode each edge by binding the vertex codes of its start and end vertices.
    \item Represent the edgeset as a superposition, or sum, of its edge encodings.
\end{enumerate}
We propose the tensor-spherical embedding, which uses the tensor product as the binding operation and random spherical vectors as the vertex code. In this paper, we highlight unique properties of the tensor product that make it an especially attractive binding operation for graph embeddings.

\section{Related Work}
% Bind-and-sum methods broadly fall under the field of hyperdimensional computing (HDC) and vector symbolic architectures (VSAs) \cite{kanerva_sdm} \cite{kleyko_VSA}\cite{Gayler2009ADB} \cite{PlateHRR}. Both seek to imbue vector spaces with the ability to represent symbolic structures, and they usually consist of an encoding operator and binding operator. Objects are encoded as vectors through the former, and the symbolic binding of objects is represented by applying the binding operator to their vector codes. For example, after encoding the concepts \textit{black} and \textit{cat} as two vectors $u$ and $v$, we can represent the tuple $(black,cat)$ via $u \otimes v$ where $\otimes$ is the binding operator. Sets of objects are then represented by the superposition, or sum, of their vector codes. Going back to the previous example, we can represent the set $\set{(black,cat),(brown,dog)}$ as $u \otimes v + x \otimes y$, where \textit{brown} and \textit{dog} are encoded as $x$ and $y$ respectively.
Bind-and-sum methods broadly fall under the field of hyperdimensional computing (HDC) and vector symbolic architectures (VSAs) \cite{kanerva_sdm} \cite{kleyko_VSA}\cite{Gayler2009ADB} \cite{PlateHRR}. Both seek to imbue vector spaces with the ability to represent symbolic structures, and they usually consist of an binding operator and bundling operator. Objects are first encoded as vectors, and the symbolic binding of objects is represented by applying the binding operator to their vector codes. Sets of objects are then represented by the bundling operator. In a majority of cases, this bundling operator is summation, and we represent a set of objects as the sum, or superposition in HDC, of its object embeddings \cite{kanerva_sdm} \cite{kleyko_VSA}. This is called the superposition principle, and we establish a link between superposition and the tensor product.

The tensor product was first proposed by Smolensky \cite{Smolensky_tensor} as a binding operation for representing general role-filler pairs. We study the tensor product in the context of graphs where the binding order is fixed and small, sidestepping a major issue where the tensor product's dimension explodes with number of bound objects. Similarly, pseudo-orthogonal high-dimensional vectors have been proposed and explored by previous work \cite{kanerva_sdm} \cite{kleyko_VSA} \cite{Cheung2019SuperpositionOM} \cite{Thomas2020ATP} and are a core component of HDC. Spherical codes are one case of pseudo-orthogonal random vectors, and we study their synergistic pairing with the tensor product. Other works have also analyzed the capacities of different VSA schemes \cite{clarkson2023capacity} and have given experimental comparisons between different schemes \cite{Schlegel_2021}. We restrict our analysis to the context of graphs and analyze graph operations that do not have analogues in a general VSA, such as edge composition.

As mentioned previously, the proposed graph embedding is a bind-and-sum method, and many similar graph embedding methods have been proposed which use different binding and coding operations. Some use the circular correlation as a binding operation \cite{Nickel2016HolographicEO} 
 \cite{Ma2018HolisticRF}, while others use the Hadamard product \cite{Poduval_graph_embed} \cite{Kang2022RelHDAG} \cite{Nunes2022GraphHDEG}. Many such methods use learned vertex codes \cite{Nunes2022GraphHDEG} \cite{Nickel2016HolographicEO} \cite{Kang2022RelHDAG}, but in this paper we restrict ourselves to a random spherical codes. While some graph embedding works use tensor factorization to deconstruct the adjacency tensor \cite{ZulaikaTensor} \cite{NickelTensor} \cite{nickel2013logistic}, we have been unable to find another bind-and-sum graph embedding that use the tensor product as a binding operation. This might be due to its perceived memory cost relative to compressed bindings like the Hadamard product in the HDC community, but in this paper we shall challenge this perception.

% Given the rich literature surrounding this problem, we believe our contribution is two-fold. Firstly, we provide a detailed analysis of the capabilities of tensor-spherical embeddings, as well as the tensor product and spherical codes individually. Interesting results include establishing a link between the tensor product and superposition, showing it is the most general therefore most expressive binding operation that respects superposition. Secondly, our analysis of the memory-capacity ratio and comparison to other bind-and-sum methods is novel, and indeed our results challenge common objections to the tensor product as too memory inefficient. We show that compressed binding operations suffer a proportionate penalty in their embedding's capacity, suggesting that memory-efficient alternatives to the tensor product don't actually offer any actual savings while sacrificing the expressivity of the tensor product. In conclusion, our contribution is our detailed analysis of the tensor-spherical embedding and, in particular, the many nice properties of the tensor product as a binding operation for graph embeddings.

\section{Notation and Terminology}
Our work focuses on embedding objects - vertices, edges, graphs - into vector spaces. Therefore, we shall denote the objects using bolded letters and their embeddings with the corresponding unbolded letter. For example, we denote a graph using $\boldsymbol{G}$ and its embedding using $G$. In later sections, when there is no confusion we shall drop the distinction between the object and its embedding, using just the unbolded letter.

We also adopt graph terminology that may be non-standard for some readers. There are many names for the two vertices of a directed edge $(d,c) = d \rightarrow c$.  In this paper, we shall call $d$ the domain of the edge and $c$ the codomain. Unless stated otherwise, all graphs will be directed graphs.
\section{Embedding Method Overview}
Given some set $\boldsymbol{V}$, we study a method for embedding the family of graphs that can be made from $\boldsymbol{V}$. Specifically, for any directed graph $\boldsymbol{G} = (\boldsymbol{V_G}, \boldsymbol{E_G})$ such that $\boldsymbol{V_G} \subseteq \boldsymbol{V}$ and $\boldsymbol{E_G} \subseteq \boldsymbol{V} \times \boldsymbol{V}$,  we embed $\boldsymbol{G}$ by embedding its edge-set $\boldsymbol{E_G}$ into a vector space. To this end, we first assign each vertex in $\boldsymbol{V}$ to a $d$-dimensional unit vector, drawn independently and uniformly from the unit hypersphere $\mathbb{S}^{d-1}$. We then embed each directed edge $(\boldsymbol{d},\boldsymbol{c})$ in the edgeset by the tensor product of their vertex embeddings:
\[
(\boldsymbol{d}, \boldsymbol{c}) \mapsto d \otimes c
\]
Fixing a basis, the tensor product is the outer product $d c^T$. Then, the embedding $G$ of the graph $\boldsymbol{G}$ - which we represent by its edgeset $\boldsymbol{E_G}$ - is the sum of its edge embeddings:
\[
\boldsymbol{G} \mapsto G = \sum_i d_i \otimes c_i = \sum_i d_i c_i^T
\]
Our embedding method can be separated into two parts: the spherical code used to embed vertices and the tensor product used to bind vertices into edges. We shall analyze each part separately as well as their synergistic pairing.

Finally, while we represent a graph $\boldsymbol{G}$ by representing just its edgeset $\boldsymbol{E_G}$, we can also embed its vertex set by augmenting $\boldsymbol{E_G}$ with self-loops $(\boldsymbol{v},\boldsymbol{v})$ for each $\boldsymbol{v} \in \boldsymbol{V_G}$. This introduces confusion between a vertex and its self-loop, but for certain graphs,such as directed acyclic graphs, this is not a concern. However, we will focus on the original method of solely embedding the edgeset in this paper. 

\section{Graph Operations} \label{GraphOp}
In this section, we will give an overview of some core graph operations possible under the proposed embedding framework. One key property we require is that the vertex codes be (nearly) orthonormal vectors, and this is the primary motivation behind our use of spherical codes. Later sections will give a more precise analysis of near orthonormality, but for clarity we assume exactly orthonormal vertex codes in this section.

\subsection{Edge Addition/Deletion}
Adding/deleting the edge $(a,b)$ to a graph corresponds to simply adding/subtracting $ab^T$ from the graph embedding $G$.

\subsection{Outbound/Inbound Vertices}
For a given vertex $d$ in $G$, suppose we wanted to find all the vertices $c_i$ that a vertex $d$ points to: $(d,c_i)$. To do this, we multiply $G$ by $d$ on the left:
\[ d^T G = d^T (\sum_i d_i c_i^T) = \sum_i \dprod{d}{d_i} c_i^T = \sum_{d_i = d} c_i^T
\]
In the last equality, we use the fact that the vertex embeddings are orthonormal. Since we represent sets as sums, this result is precisely the set of vertices that $d$ points to.

Conversely, say we were interested in finding all vertices $d_i$ that point to a vertex $c$: $(d_i,c)$. This would analogously correspond to right multiplying by $c$:
\[ G c = (\sum d_i c_i^T) c = \sum \dprod{c_i}{c} d_i = \sum_{c_i = c} d_i
\]
We can generalize this to a set of candidate vertices. For example, we first represent a set of vertices $S$ by the sum of its constituents $S = s_1 + \cdots + s_n$. Then, to find all vertices in $G$ that are connected to any vertex in $S$ by an outbound edge, we multiply $G$ on the right by $S$:
\[ G S = (\sum d_i c_i^T)(s_1 + \cdots + s_n) = \sum_{c_i \in S} d_i
\]
Analogously, to find all vertices in $G$ connected to $S$ by an outbound edge, we multiply $G$ on the left by $S^T$.
All subsequent operations are also multilinear like the edge query, and so they can all be similarly extended to handle sets of vertices or edges. For simplicity, we will focus on just the singleton case from now on.`

\subsubsection{Edge/Node Queries}
Combining the two operations above, checking if a graph $G$ contains the edge $(d,c)$ would correspond to multiplying $G$ by $d$ on the left and $c$ on the right:
\[ d^T G c = d^T (\sum d_i c_i^T) c = \sum \dprod{d}{d_i} \dprod{c}{c_i} = 1_{\set{(d,c) \in E_G}}
\]
Now suppose $G$ were augmented with self-loops for each vertex, where $G$ contains the loop $v v^T$ for each of its vertices $v$. One can use the above procedure to detect if $G$ contains a given vertex.

\subsection{Edge Composition, k-length Paths, and Graph Flow}

We can also perform edge composition with the graph embeddings. For example, given edges $(a,b)$ and $(b',c)$, we want to return the composite edge $(a,c)$ only if $b = b'$. This can be done by matrix multiplication of the edge embeddings:
\[ (ab^T)(b'c^T) = \dprod{b}{b'} a c^T = 1_{\set{b = b'}}a c^T
\]
Therefore, we can compose all edges in a graph $G$ by computing its second matrix power:
\[ G^2 = (\sum d_i c_i^T)(\sum d_j c_j^T) = \sum_{i,j} \dprod{c_i}{d_j} d_i c_j^T = \sum_{c_i = d_j} d_i c_j^T
\]
Generalizing to paths of length $k$, the matrix powers of $G$ correspond to the sets of paths of a fixed length:
\[ G^k = \set{\text{all $k$-length paths in $G$}}
\]
We can combine this with vertex queries to compute graph flows. Say one wanted to know all vertices $c$ that are reachable from starting vertex $d$ by a path of length $k$. Then one would compute:
\[ 
d^T G^k = d^T (\sum d_i c_i^T)(G)^{k-1}  =  \sum_{d_i = d} c_i^T (G^{k-1}) = \cdots = \set{\text{$c$ such that there is a $k$-length path $d \rightarrow c$}}
\]
An analogous operation exists for determining all vertices $d$ that end up at final vertex $c$ after a path of length $k$:
\[
G^k c = \set{\text{$d$ such that there is a $k$-length path $d \rightarrow c$}}
\]
\subsection{Edge Flipping, Undirected Graphs, and Alternization}
Given an edge $(a, b)$ with representation $ab^T$, the representation of the flipped edge $(b, a)$ is the transpose matrix:
\[ b a^T = (a b^T)^T
\]
Therefore, the embedding of the dual graph $G^{op}$ - flipping all edges of $G$- is the transpose of $G$.
\[
G^{op} = G^T
\]
One interesting application of this property is to extend our current graph framework to undirected graphs using the un-normalized matrix symmetrization procedure:
\[
G \mapsto G + G^T
\]
This procedure may introduce double-counting of certain edges, but for certain graphs, like directed acyclic graphs, this is not an issue.

Symmetrization corresponds to making a graph undirected, and alternization also has a graph interpretation. Recall that the un-normalized alternization procedure is:
\[
G \mapsto G - G^T
\]
Alternization sends an edge $ab^T$ to $a \wedge b = ab^T - ba^T$, and now when we take the transpose of the alternating edge $a \wedge b$ we have $(a \wedge b)^T = -(a \wedge b)$. Edges are now signed, where a negative sign reverses the direction of the edge. This leads to some interesting implications when interpreting coefficients of edges. For example, subtracting an edge now corresponds to adding the opposite edge rather than just edge deletion, and this in turn gives a natural interpretation to the zero coefficient: an equal quantity of the forward and reverse edges. This in turn gives allows one to interpret a coefficient's magnitude as the net flow of some quantity, while the coefficient's sign represents the flow's direction.

\subsection{Subsetting and Subgraphs}
In a previous section, we saw that vertex queries can determine the vertices connected to or from a set of vertices $\boldsymbol{S} \subseteq \boldsymbol{V}$. Now, suppose we wanted to know the edges whose domain is in $\boldsymbol{S}$. Abusing notation, let $S$ denote the matrix whose columns are the vertex embeddings $s \in S$, and let $P_S = S S^T$ be the associated projection matrix. To find all edges whose domain is in $S$, we compute:
\[ P_S G = S S^T (\sum d_i c_i^T) = \sum \dprod{s_j}{d_i}s_j c_i^T = \sum 1_{\set{s_j = d_i}} s_j c_i^T = \sum_{d_i \in S} d_i c_i^T
\]
Similarly, to find all edges whose codomain is in $S$, we compute:
\[ 
G P_S = (\sum d_i c_i^T) S S^T = \sum \dprod{s_j}{c_i} d_i s_j^T = \sum 1_{\set{s_j = c_i}} d_i s_j^T = \sum_{c_i \in S} d_i c_i^T
\]
Now, say we wanted to extract the full subgraph $G_S$ of $G$, whose edges are those of $G$ that have both domain and codomain in $S$. Combining the above two equations, this amounts to conjugating $G$ with $P_S$:
\[ G_S = P_S G P_S
\]

\subsubsection{Translation between Vertex Codes} \label{s2.5.1}
Suppose we have two different vertex codes $\phi_1: V \rightarrow V_1$ and $\phi_2: V \rightarrow V_2$. For a graph $G$, we then have a graph embedding induced from each vertex code, denoted $\phi_1(G)$ and $\phi_2(G)$. There is a natural way to convert $\phi_1(G)$ to $\phi_2(G)$. Let $\Phi_1$ be the matrix who columns are the vertex codes $\set{\phi(v_1), \phi(v_2), \cdots}$ in that order; let $\Phi_2$ be the similarly defined matrix for $\phi_2$. Then, we have the transition map between vertex codes $\Psi: V_1 \rightarrow V_2$ represented by the matrix $\Phi_2 \Phi_1^T$:
\[ \Psi: V_1 \rightarrow V_2 \quad ; \quad \phi_1(v_i) \mapsto \phi_2(v_i) \quad ; \quad \Psi = \Phi_2 \Phi_1^T
\]
Then, the graph representations $\phi_1(G)$ and $\phi_2(G)$ follow the relation:
\[ \phi_2(G) = \Psi \phi_1(G) \Psi^T
\]
% Note that in the setup above, we implicitly assumed that each vertex code was unique, so there was a bijection between codes $\phi_1$ and $\phi_2$. However, we can generalize the above procedure to the case where the vertex code is procedure allows us to naturally transfer the graph embeddings of one code to a new code. For example, suppose the vertex set $V$, with code $\phi_1(V)$, was equipped with an equivalence relation, and let us denote the equivalence classes of $V$ as $[v]$. For a graph $G$ in $V$, say we wanted to compute the induced graph where we collapse nodes into their equivalence classes. Then, our trnsit

\subsection{Counting via the Trace}
The trace of a graph embedding $G$ will count the number of self-loops in the graph:
\[ tr(G) = tr( \sum d_i c_i^T) = \sum \dprod{d_i}{c_i} = \sum 1_{\set{d_i = c_i}}
\]
This property can be used for many interesting graph operations.
\subsubsection{Vertex Counting in Augmented Graphs}
For a graph $G$ whose edgeset set is augmented with the self-loops $vv^T$ for every $v$ in its vertex set $V_G$, the trace will naturally return the cardinality of its vertex set:
\[ tr(G) = tr( \sum d_i c_i^T) = \sum \dprod{d_i}{c_i} = |V_G|
\]
\subsubsection{Edge Counting and a Natural Metric}
To count the number of edges in a graph $G = \sum d_i c_i^T$, one computes:
\[ tr(G^T G) = tr((\sum c_i d_i^T)(\sum d_j c_j^T)) = tr(\sum \dprod{d_i}{d_j} c_i c_j^T) = \sum \dprod{d_i}{d_j}\dprod{c_i}{c_j} = |E_G|
\]
This gives a nice relation to the Frobenius norm $||\cdot||_F$ via the identity $||A||_F^2 = tr(A^T A)$. This shows that the squared Frobenius norm of a graph representation is precisely the cardinality of its edge set. Hence, the natural metric on our graph embeddings would be the one induced from the Frobeinus norm:
\[
d(G,H) \coloneqq ||G-H||_F = \sqrt{\# \text{ of different edges}}
\]
where this metric depends solely on the number of different edges between the two graphs.

\subsubsection{Testing Graph Homomorphisms}
Finally, one interesting application of the trace is to quantify  the 'goodness' of a proposed graph homomorphism. Recall that a graph homomorphism $f$ from graph $G$ to graph $H$ is comprised of a vertex function $f_1 : V_G \rightarrow V_H$ and edge function $f_2: E_G \rightarrow E_H$ such that every edge $(a, b)$ is mapped by $f_2$ to the edge $(f_1(a), f_1(b))$. Hence, every homormorphism is completely determined by its vertex function, and so it suffices to test if a vertex function $f: V_G \rightarrow V_H$ induces a graph homomorphism. Therefore, for every proposed vertex function $f$ we assign a quantity called the graph homomorphism coefficient, described in detail below.

Given two graphs $G =  \sum d_i c_i^T$ and $H = \sum a_j b_j^T$ with proposed vertex function $f: V_G \rightarrow V_H$, the vertex function $f$ induces the following map:
\[ G = \sum d_i c_i^T \mapsto \sum f(d_i) f(c_i)^T = f(G)
\]
Note that $f(G)$ can be computed using the method mentioned in the previous section. If $f$ is truly a graph homomorphism, then $f(G)$ is a subgraph of $H$ and every edge of $f(G)$ is an edge of $H$. We therefore compute the number of matching edges between $f(G)$ and $H$, denoted $\delta_f$:
\[ \delta_f = tr(f(G)^T H) = tr(( \sum f(c_i) f(d_i)^T)(\sum a_j b_j^T)) = \sum \dprod{f(d_i)}{a_j} \dprod{f(c_i)}{b_j}
\]
Dividing $\delta_f$ by $|E_G|= tr(G^T G)$ gives the fraction of edges in $f(G)$ that have matches in $H$. We call this ratio the graph homomorphism coefficient, denoted as $\Delta_f \coloneqq \frac{\delta_f}{|E_G|}$. Note that $\Delta_f \in [0,1]$. If $f$ truly induces a graph homomorphism, then every edge has a match and $\Delta_f = 1$; conversely, if $f(G)$ has no matching edges in $H$, then $f(G)$ is totally different from $H$ and $\Delta_f = 0$. In this sense, the graph homomorphism coefficient gives a measure of how close a vertex function $f$ is to inducing a graph homomorphism.

\subsection{Vertex Degree and Gram Matrices}
In a directed graph $G$, consider a fixed vertex $v$. The in-degree of $v$, $in(v)$, is the number of edges that go into $v$ - ie. have codomain $v$. Similarly, the out-degree of $v$, $out(v)$, is the number of edges that go out of $v$ - ie. have domain $v$. Generalizing these definitions, we define the joint in-degree of two vertices $v$ and $w$, $in(v,w)$, as the number of vertices $d_i$ that have both an edge into $v$ and an edge into $w$. The joint out-degree of two vertices $v$ and $w$, $out(v,w)$, is similarly defined. Note that using the generalized definitions, the in/out-degree of a vertex with itself coincides with the single-case definition. Using our graph embeddings, there is a natural way to compute both of these degrees.

We first handle the in-degree of a vertex. Recall that $Gv$ is the superposition, or sum, of all vertices $d_i$ that go into $v$ with an edge $(d_i,v)$. Since vertex codes are orthonormal, we can use the squared Euclidean norm to count the number of vertices in superposition:
\[
||Gv||^2 = (Gv)^T Gv = v^T (G^T G) v = in(v)
\]
In fact, for two vertices $v$ and $w$, we again appeal to orthonormality of the vertex codes to compute their joint in-degree:
\[
(Gv)^T Gw = v^T (G^T G) w = in(v,w)
\]
Thus, we see that the Gram matrix $G^TG$ represents a bilinear form that computes the joint in-degree of two vertices:
\[
G^TG: V \times V \rightarrow \R \qquad; \qquad v \times w \mapsto v^T (G^TG) w = in(v,w)
\]
We shall called the gram matrix $G^T G$ the in-degree matrix.

A similar line of reasoning shows that the Gram matrix $GG^T$ computes the joint out-degree between two vertices. Hence, we shall call $G G^T$ the out-degree matrix. Immediately, one natural question when working with Gram matrices are their spectral properties. Indeed, there are natural links to properties of graph connectivity, which we discuss in the next section.

\section{Superposition and the Tensor Product}\label{TensorDeriv}
% In this section, we derive the tensor product binding from superposition principle and show it is the most general embedding method. We also give an alternative derivation of the tensor product and orthonormality when considering suitability toward certain graph operations.

% \subsection{Superposition and the Tensor Product}
We shall derive the tensor product from the superposition principle used in hyperdimensional computing schemes, which states that a set's vector code is the sum of its members' vector codes. Indeed, along the way we shall establish the following connection between the superposition principle and the tensor product:
\begin{theorem}
For any fixed vertex code $\phi: \boldsymbol{V} \rightarrow V$, let $\psi: V \times V \rightarrow \R^n$ be any binding operation that respects superposition. Then the resulting bound code $\psi(V,V)$ has a unique linear derivation from the tensor product $V \bigotimes V$. In this sense, the tensor product is the most general binding operation.
\end{theorem}
\begin{proof}
Assume we are given a vertex code $\phi: \boldsymbol{V} \rightarrow V$. Given a directed graph $G = (V_G, E_G)$ in $V$, our task is to represent $G$ by embedding edge set $E_G$ via a superposition of its edges. First, suppose $G$ has multiple edges from a common domain $\{d\}$ to multiple codomains $\{c_1, \dots, c_k\}$:
\[ 
\set{(d,c_1), \cdots, (d,c_k)} = \bigcup_i \set{(d,c_i)}
\]
Similarly, consider the reverse situation to multiple domains with a single codomain. Then, any edge embedding $\psi:V \times V \rightarrow \R^n$ that respects superposition must satisfy the following two equations:
\begin{gather*}
    \psi(\set{(d,c_1), \cdots, (d,c_k)}) = \sum \psi(\set{(d, c_i)})\\
    \psi(\set{(d_1, c), \cdots, (d_{l}, c)}) = \sum \psi(\set{(d_j, c)})
\end{gather*}
Thus, $\psi$ induces a bilinear function $\Tilde{\psi}$ on the vertex code via: \begin{gather*}
    \Tilde{\psi}: V \times V \rightarrow \R^n \\
    (\phi(v),\phi(w)) \mapsto \psi(v,w)
\end{gather*}
By the universality of the tensor product \cite{lang02}, the bilinear $\Tilde{\psi}$ map has a corresponding unique linear map $\psi^*: V \bigotimes V \rightarrow \R^n$, and so the tensor product uniquely maps into every edge embedding that respects superposition. 
\end{proof}
For example, three common alternative binding operations - Hadamard product, convolution, and circular correlation - all are linearly induced from the tensor product.  Representing the tensor product as the outer product matrix, the Hadamard product is the diagonal of the matrix while the convolution and circular correlation are sums along pairs of diagonals. Therefore, given a vertex code $\phi$ we may as well consider its natural induced edge embedding under the tensor product.

\section{Spherical Codes and Approximate Orthonormality}
In previous sections, we assumed our vertex codes to be exactly orthonormal. However, this orthonormality requirement is not efficient in the number of codes we can pack into a space, since we can pack at most $d$ orthonormal vectors into a $d$-dimensional space. Instead, we relax the requirement of strict orthonormality and use approximately orthornomal spherical codes. Two natural questions arise: how many pseudo-orthonormal vectors can one pack into $\R^d$, and how close do spherical codes come to achieving this limit? In this section we shall answer these two questions and give an account of approximate orthonormality in high dimensions.

\subsection{Packing Upper Bounds}
First, we give an upper bound on the number of approximately orthornormal vectors one can pack in $\R^d$. To make approximate orthonormality precise, we say the unit vectors $u,v$ are \textit{$\epsilon$-orthogonal} if $\abs{\dprod{u}{v}} < \epsilon$. The Johnson-Lindenstrauss Theorem \cite{JLLemma} implies a packing upper bound, which is shown to be tight \cite{JLopt}.
\begin{theorem}[Johnson-Lindenstrauss] \label{JL Theorem}
Let $x_1, \cdots, x_m \in \R^N$ and $\epsilon \in (0,1)$. If $d > \frac{8 ln(m)}{\epsilon^2}$, then there exists a linear map $f: \R^N \rightarrow \R^d$ such that for every $x_i, x_j$:
\[ (1-\epsilon)||x_i-x_j||^2 \leq ||f(x_i) - f(x_j)||^2 \leq (1+\epsilon) ||x_i - x_j||^2
\]
\end{theorem}
\begin{corollary} \label{PackingUpperBound}
For any $\epsilon \in (0,\frac{1}{2})$, one can pack at most $O(e^{C d \epsilon^2})$ $\epsilon$-orthogonal unit vectors in $\R^d$ for some universal constant $C$.
\end{corollary}
\begin{proof}
Let $m$ be any integer such that $d > \frac{8 ln(m+1)}{\epsilon^2}$. Consider the set of $(m+1)$ vectors $\set{e_1, \cdots, e_m, 0}$ in $\R^m$. These satisfy the conditions of the Johnson-Lindenstrauss Theorem, and let $f$ be the JL map. Then,
\begin{align*}
    (1-\epsilon)=(1-\epsilon)||x_i||^2 \leq ||f(x_i) - 0 ||^2 = ||f(x_i)||^2 \leq (1+\epsilon)
\end{align*}
Hence, since $||e_i - e_j||^2 = 2$ for $i \neq j$:
\begin{align*}
    2(1-\epsilon) \leq ||f(e_i) - f(e_j)||^2 = ||f(e_i)||^2 + ||f(e_j)||^2 - 2 \dprod{f(e_i)}{f(e_j)} \leq 2(1+\epsilon)
\end{align*}
Plugging the first equation to the second, we have:
\[ 
-2\epsilon \leq \dprod{f(e_i)}{f(e_j)} \leq 2\epsilon
\]
Thus, using $u_i = \frac{f(e_i)}{||f(e_i)||}$, the first equation, and the fact that $\epsilon < \frac{1}{2}$:
\[ -4 \epsilon \leq -\frac{2\epsilon}{(1-\epsilon)} \leq \dprod{u_i}{u_j} \leq \frac{2\epsilon}{1-\epsilon} \leq 4 \epsilon
\]
Therefore we have:
\[ D > 8 ln(m+1)/\epsilon^2 \implies m = O(e^{Cd\epsilon^2})
\]
for some constant $C$.
\end{proof}

\subsection{Spherical Codes}
We use spherical codes, which are random vectors drawn uniformly from the $d$-dimensional unit hypersphere $\mathbb{S}^{d-1}$, to generate approximately orthornormal vertex codes. In this section,  we give the exact distribution of their dot product and show that they achieve the packing upper bound in Corollary \ref{PackingUpperBound}.

\subsubsection{Distribution of the Dot Product}
\begin{theorem}[Dot Product of Spherical Codes] \label{DP Dist}
If $u,v$ are uniformly distributed over $\mathbb{S}^{d-1}$, let $X = \dprod{u}{v}$ and $Y = \frac{X+1}{2}$. Then, $Y$ follows a Beta$(\frac{d-1}{2}, \frac{d-1}{2})$ distribution.
\end{theorem}
\begin{proof}
Since both $u$ and $v$ are uniformly distributed over $\mathbb{S}^{d-1}$, by symmetry we may fix $v$ as any unit vector without changing the distribution of $X$. Hence, let $v$ be the first coordinate vector $e_1$, and so $X = \dprod{u}{e_1} $. The set of vectors $u$ such that $\dprod{u}{e_1} = x$ form a spherical section of $\mathbb{S}^{d-1}$: a $(d-1)$-sphere of radius $\sqrt{1-x^2}$. Then, the set of vectors with dot product $ x \leq X \leq x + \delta$ corresponds to a $d$-dimensional belt on the sphere. Since the probability of a set is proportional to its surface area, $P(x \leq X \leq x + \delta)$ is proportional to the surface area of this belt. The area of a $d$-dimensional sphere of radius $r$ is $C r^{d-1}$ for some constant $C$, and if $\theta = \cos^{-1}(X)$ we have:
\begin{align*}
P(x \leq X \leq x + \delta) &\propto \int_{\cos^{-1}(x)}^{\cos^{-1}(x + \delta)} (\sqrt{1-cos^2(\theta)})^{d-2} d\theta\\
&= \int_x^{x + \delta} (\sqrt{1-t^2})^{d-2} d(\cos^{-1}{t}) \propto \int_x^{x + \delta} (1-t^2)^{\frac{d-3}{2}} dt
\end{align*}
Hence, if $f_X$ is the density of $X$, $f_X \propto (1-x^2)^{\frac{d-3}{2}}$. Then, letting $X+1 = 2Y$, we can simplify:
\[
(1-x^2)^{\frac{d-3}{2}} = (1-x)^{\frac{d-3}{2}}(1+x)^{\frac{d-3}{2}} = (2 - 2y)^{\frac{d-3}{2}}(2y)^{\frac{d-3}{2}} \propto (1-y)^{\frac{d-3}{2}}y^{\frac{d-3}{2}}
\]
Thus, $Y = \frac{X+1}{2}$ follows a Beta$(\frac{d-1}{2}, \frac{d-1}{2})$ distribution
\end{proof}
\begin{corollary}\label{Dprod moments}
The dot product $X$ between $u,v$, uniformly distributed over $\mathbb{S}^{d-1}$, has $E(X) = 0$ and $Var(X) = \frac{1}{d}$.
\end{corollary}
\subsubsection{Optimality of Spherical Codes}
Now, given a set of $k$ spherical codes, we want to ensure that all of them are approximately orthornomal. Surprisingly, we shall see that for a fixed probability of violating $\epsilon$-orthonormality, the number of spherical codes $k$ will be $O(e^{Cd\epsilon^2})$ for some universal constant $C$, matching our upper bound.
\begin{theorem}\label{Optimality}
For $k$ vectors $x_1, \cdots, x_k$ sampled iid from the uniform distribution on $\mathbb{S}^{d-1}$ and for any $\epsilon > 0$, we have:
\[ \max_{i \neq j} | \dprod{x_i,x_j} | < \epsilon
\]
with probability at least $1 - 2\binom{k}{2}e^{-\frac{d}{8}\epsilon^2 }$
\begin{proof}
Let $X$ denote the dot product between two different vectors. From Theorem \ref{DP Dist}, $\frac{X+1}{2}$ follows a $Beta(\frac{d-1}{2},\frac{d-1}{2})$ distribution, which is strictly sub-Gaussian with parameter $\frac{1}{2\sqrt{d}}$ \cite{MarchalBetaSub}. Hence, applying the standard sub-Gaussian concentration inequality \cite{Boucheron2004}:
\[
    P( X > \epsilon) = P(\tfrac{X+1}{2} - \tfrac{1}{2} > \epsilon/2) \leq e^{-d \epsilon^2/8}
\]
Since $X$ is symmetrically distributed, the lower tail bound is the same. Combining the two and using a union bound over all $\binom{k}{2}$ pairs gives the result.
\end{proof}
\end{theorem}
Hence, for a fixed error threshold $T$, we can choose the maximum number of vectors $k$ such that the error probability is less than $T$:
\[ 
T \approx C_1\binom{k}{2}e^{-C_2d\epsilon^2} \implies  k \approx O(e^{C_3 d \epsilon^2})
\]
for universal constants $C_1,C_2,C_3$. This matches the upper bound given by Theorem \ref{JL Theorem}, and spherical codes achieve the packing upper bound in Corollary \ref{PackingUpperBound}.

\section{Alternative Binding Operations} \label{AlterBindOp}
As a bind-and-sum approach, our graph embedding method uses the tensor product to bind vertex codes into an edge. However, this means that the graph embedding space scales quadratically with the vertex code dimension, and in response various memory efficient alternatives have been used in the HDC literature \cite{Nunes2022GraphHDEG} \cite{Kang2022RelHDAG} \cite{Nickel2016HolographicEO}. In this section, we shall primarily analyze three prominent alternative binding operations: the Hadamard product, convolution, and circular correlation. This, in addition to Section \ref{Random Codes}, will prepare us to analyze other bind-and-sum embeddings and compare them to tensor-spherical embeddings.

\subsection{Hadamard Product}
The Hadamard product of two vectors $a$ and $b$, denoted $a \odot b$, is their element-wise multiplication:
\[
[a \odot b]_k = a_k b_k
\]
It preserves the dimension of the vertex code and scales linearly with the code dimension. The most common vertex codes when using the Hadamard product are phasor codes and binary codes: vectors whose entries are the complex and real units respectively. These codes have nice unbinding operations with respect to the Hadamard product, where multiplication by the conjugate of a vertex code will remove that code from the binding:
\[
\overline{a} \odot (a \odot b) = b
\]
In the next section, we also briefly touch on other codes. However, we shall see that since unbinding the Hadamard product requires element-wise division, many choices of random codes are not suitable for accurate graph operations.

\subsection{Convolution}
The convolution of two vectors $a$ and $b$, denoted $a * b$, is defined as:
\[
[a * b]_k = \sum_i a_i b_{k-i} 
\]
If $\mathcal{F}$ denotes the Fourier transform, then the convolution can also be expressed as:
\[
a * b = \mathcal{F}^{-1}(\mathcal{F}(a) \odot \mathcal{F}(b) )
\]
Hence, convolution and the Hadamard product are equivalent up to a Hermitian change of basis, and so it is sufficient to analyze just the Hadamard product. Indeed, common codes used with convolution are those whose Fourier transforms are phasor and binary codes, and there is a bijection between the codes used for convolution and codes used for the Hadamard product.

\subsection{Circular Correlation}
The circular correlation of two vectors $a$ and $b$, denoted $a \star b$, is defined as:
\[
[a \star b]_k = \sum_i a_i b_{k+i}
\]
Using the Fourier transform, the circular correlation also can be expressed as:
\[
a \star b = \mathcal{F}^{-1}(\overline{\mathcal{F}(a)} \odot \mathcal{F}(b))
\]
For a vector $a$, let $y_k$ be the $k^{th}$ Fourier coefficient of $a$:
\[
y_k = \sum_{s=0}^{n-1} \exp(\frac{2\pi i (-ks)}{n}) a_s
\]
Taking the conjugate, we get:
\[
\overline{y_k} =  \sum_{s=0}^{n-1} \exp(\frac{2\pi i ( ks)}{n}) a_s =  \sum_{s'=0}^{n-1} \exp(\frac{2\pi i (-ks')}{n}) a_{-s'}
\]
Hence, $\mathcal{F}^{-1}(\overline{\mathcal{F}a} )$ = $Pa$, where $P$ is the permutation that flips indices, and so the circular correlation is convolution augmented with flipping the first argument. This is a special case of using permutation to induce ordered bindings, which we shall cover in the next subsection. Since the circular correlation is a permuted convolution and convolution is equivalent to the Hadamard product, it is sufficient to study the Hadamard product with permutations.

\subsection{Alternative Bindings as Compressions of the Tensor Product}
As noted in our derivation of the tensor product from superposition, the tensor product is a universal construction. That is, every binding method that respects superposition - bilinearity - is linearly induced from the tensor product. The three alternative binding operations - the Hadamard product, convolution, circular correlation- are no exception.\\
\\
More specifically, fixing the standard basis the tensor product $v \otimes w$ is the outer product $v w^T$, a $d \times d$ matrix. The Hadamard product is the main diagonal of this matrix. Similarly, the entries of the circular correlation are sums along pairs of diagonals of $v w^T$, spaced by an interval of $2(d-1)$.  For convolution, we analogously take sums along pairs of the reverse diagonals (bottom left to top right).\\
\\
Hence, all three operations can be seen as compressions of the tensor product, where we compress a $d \times d$ matrix into a $d$-dimensional vector. While on the surface this might seem to save much in terms of memory, in the next subsection we shall see these compressed binding operations are unable to perform some fundamental graph operations. Moreover, in Sections \ref{EdgeQuer} and \ref{EdgeComp} we show that the Hadamard product's compression incurs a proportionate penalty in the number of edges it can accurately store in superposition, meaning it offers no real savings in memory.

\subsection{Hadamard Product: Graph Functionality}
As established in the preceding subsections, all three operations are, up to a Hermitian change of basis, equivalent to the Hadamard product with possibly a permutation applied to one of its arguments. Hence, in this section we shall analyze the Hadamard product and its suitability for graph embeddings.

Firstly, the Hadamard product can perform edge composition when used in conjunction with binary codes. This is because it has easy unbinding operations, where multiplication by a code will remove that code from the binding. Using Hadamard-binary embeddings, we can perform edge composition by taking the Hadamard product of two edges:
\[
(a \odot b) \odot (b \odot c) = a \odot c
\]
Moreover, note that in the case of a mismatch between the vertices, the resulting edge is:
\[
(a \odot b) \odot (b' \odot c) = a \odot c \odot n
\]
where $n$ is a noisy binary code. Importantly, there is no destruction of mismatched edges during edge composition. This will impact the number of edges we can accurately store in superposition, which we shall cover in Section \ref{EdgeComp}. 

One can ask if there is any alternative edge composition function that is better suited for the Hadamard product. Assuming this edge composition function respects superposition, one can show that the natural edge composition function under the Hadamard product is again the Hadamard product. We give a sketch of the argument: firstly, the desired composition function respects superposition, so it is a multilinear function and is completely determined by its action on some basis $\set{b_i}$. Fixing the standard basis, we impose the natural constraint:
\[
(e_i \odot e_j) \times (e_k \odot e_l) \mapsto 
\begin{cases}
e_i \odot e_l & j = k\\
0 & j \neq k
\end{cases}
\]
Since these constraints are satisfied by the Hadamard product, the edge composition function must be the Hadamard product. 

The Hadamard product is unable to represent directed edges because it is symmetric: $a \odot b = b \odot a$. A common fix is to permute one of its arguments before binding, so now we bind as $Pa \odot b$ where $P$ in some permutation matrix. Our augmented binding operation becomes:
\[
(a,b) \mapsto Pa \odot b
\]
This augmented operation is certainly able to represent directed edges, but it is now unable to perform edge composition. Disregarding oracle operations where one already knows the bound vertices, using the Hadamard product to perform edge composition results in:
\[
(Pa \odot b) \odot (Pb \odot c) = a \odot c \odot (b \odot Pb)
\]
The noise term $(b \odot Pb)$ will not cancel unless $P$ is the identity, the non-permuted case. A similar question arises whether there is a better edge composition funciton for the permutated Hadamard: assuming this edge composition function respects superposition, by multilinearity one can show no such operation exists. Intuitively, one would like to send the pair $(Pa,a)$ to the vector of all ones, but the symmetry in the Hadamard product and its compression means such a map is not possible without breaking multilinearity.

In summary, we see that the Hadamard product with binary codes sacrifices some core graph functionality: with the regular Hadamard product, one can perform edge composition but cannot represent directed edges; in the permuted case, one can represent directed edges but cannot perform edge composition. One can show that other core graph operations, like subsetting and graph homomorphisms, are also impossible under such schemes. Thus, in the Hadamard-binary case we see that it cannot represent graph operations that the tensor product can.

\subsubsection{Hadamard Product: Phasor Codes}
Phasor codes, which generalize binary codes, are even less suitable than binary codes. The natural unbinding operation for phasor codes is to multiply by the conjugate codes. However, this already makes it unsuitable for edge composition in both the non-permuted and permuted case. In the non-permuted case:
\[
\overline{(a \odot b)} \odot (b \odot c) = \Bar{a} \odot c \neq a \odot c 
\]
The permuted case has a similar deficiency. Again due to the compression of the Hadamard product, there is no multilinear function that conjugates just one argument of the bound edge $a \odot b$. Intuitively, this is due to the symmetry of the Hadamard product, since it unable to distinguish which particular vertex to conjugate. Edge composition using phasor codes is impossible, and any code which can be derived from the phasor code suffers a similar defect.

\subsubsection{Hadamard Product: Continuous Codes}
Similarly, since unbinding the Hadamard product requires element-wise division, many choices of random continuous codes are numerically unstable. In fact, the next section we shall see that specific cases of continuous codes all suffer from having infinite moments, making accurate graph operations impossible since the noise terms will overwhelm the signal. Morever, operations like edge composition are also impossible for similar reasons as the phasor code, where one needs invert a specific vertex of the bound edge. Due the Hadamard product's compression, such a multilinear map does not exist.

\section{Random Codes} \label{Random Codes}
In this section, we state and establish some basic results of various random codes in preparation of our analysis of other bind-and-sum approaches. Throughout this section we assume all vertex codes have common dimension $d$.

\subsection{Spherical Codes} \label{Sphere Codes}
We generate spherical codes by sampling iid from the $d$-dimensional unit hypersphere $\mathbb{S}^{d-1}$. They have the following properties:
\begin{theorem}[Spherical Code Properties]\label{SpherCodeProp}
Let $X$ denote the dot product between two spherical codes. Then, the following statements hold:
\begin{enumerate}
    \item $\frac{X+1}{2} \sim Beta(\frac{d-1}{2}, \frac{d-1}{2})$
    \item $E(X) = 0$ and $Var(X) = \frac{1}{d}$
    \item $|X| \propto \frac{1}{\sqrt{d}}$ with high probability.
\end{enumerate}
\begin{proof}
The first two claims follow from the results of Theorem \ref{DP Dist} and Corollary \ref{Dprod moments}. The third claim follows from the same sub-Gaussian concentration inequality used in Theorem \ref{Optimality}.
\end{proof}
\end{theorem}
\subsection{Rademacher Codes}\label{Rad Codes}
Rademacher codes are vectors $v$ where each entry is an iid Rademacher random variable: $v_i = \pm 1$ with probability $\frac{1}{2}$. They have the following properties.
\begin{theorem}[Rademacher Code Properties] \label{RadCodeProp}
Let $X$ denote the dot product of two Rademacher vectors. The following statements hold:
\begin{enumerate}
    \item $\frac{X+d}{2} \sim Binomial(d,\frac{1}{2})$
    \item $EX = 0$ and $Var(X) = d$
    \item $|X| \propto \sqrt{d}$ with high probability.
\end{enumerate}
\begin{proof}
Since the product of two Rademacher random variables, $X$ is the sum of $d$ Rademacher random variables. If $B$ = Bernoulli($\frac{1}{2}$), then $2B-1$ is a Rademacher random variable. Hence, the sum of $d$ Rademachers is $2 Y -d$ where $Y \sim Binomial(d,\frac{1}{2})$ and so $X = 2Y-d$. This establishes the first statement, and the second statement follows from the first. The third claim follows from Chernoff's inequality \cite{Boucheron2004}.
\end{proof}
\end{theorem}
% Note that all Rademacher codes have norm $d$. Thus, we may scale them by $\frac{1}{\sqrt{d}}$ to normalize them, and in this sense they are special cases of spherical codes. In fact, by a similar argument as the spherical codes case, one can show that for a fixed error threshold of violating $\epsilon$-orthogonality, normalized Rademacher codes achieve the Johnson-Lindenstrauss upper bound.
% However, one can have at most $2^d$ unique codes, while any finite number of spherical codes have probability zero of having a repeat. Rademacher codes still have a hard packing limit, but the trade-off is cleaner unbinding with respect to the Hadamard product.

\subsection{Other Continuous Codes}
Here, we will briefly describe three common continuous codes: Gaussian, Cauchy, and uniform vectors. The entries of Gaussian vectors are iid Gaussian, usually the standard Gaussian. Cauchy and uniform vectors are analogously generated, with their entries being drawn from the standard Cauchy and $Unif[0,1]$ respectively. The main problem with these codes is that unbinding the Hadmard product requires element-wise division, and the resulting ratios of random variables will have infinite moments. This results in ill-controlled noise terms and makes them unsuitable for accurate graph operations.
\begin{theorem}
For $t,u$ iid Gaussian, Cauchy, or uniform. Let $Y = \frac{t}{u}$. Then, for any of the three distributions, all moments of $Y$ are undefined.
\begin{proof}
In the Gaussian case, the ratio of two independent standard Gaussians is a Cauchy random variable, which is known to have infinite moments. In the Cauchy case, the ratio of two  independent standard Cauchy rv's has the density \cite{CauchyRatio}:
\[
f_{Y_c}(y) \propto \frac{1}{(y^2 - 1)}ln(y^2)
\]
Then, comparing integrals:
\[
\infty =\frac{1}{2} \int_0^{c}  |ln(y^2)| \leq \int_0^{c} y f_{Y_c} (y) \leq \int_0^\infty y f_{Y_c} (y)
\]
we see that the first moment is also undefined (for some sufficiently small constant $c$). In the uniform case, the ratio of two independent $U[0,1]$ rv's has the following density \cite{UniformRatio}:
\[
f_{Y_u}(y) = 
\begin{cases}
\frac{1}{2} & 0 < y < 1\\
\frac{1}{2z^2} & y \geq 1\\
0 & y \leq 0
\end{cases}
\]
A similar comparison test also shows that the first moment is undefined. Thus, the first moment, and hence all moments, are undefined for all three choices of distribution.
\end{proof}
\end{theorem}
% \\
% Now, for $t,u,v$ iid, let $Z = t \frac{u}{v}$. By Holder's inequality in the $(1,\infty)$ case,
% \[
% \frac{E |\frac{u}{v}|}{\alpha} \leq E|1_{|t| < \alpha}t \frac{u}{v}| = E_{|t| < \alpha} |Z| \leq E|Z|
% \]
% for any $\alpha > 1$. Hence, in all three case all moments of $Z$ are undefined.

\section{Binding Comparison: Edge Queries} \label{EdgeQuer}
In Section \ref{AlterBindOp}, we showed the the convolution and circular correlation are special cases of the Hadamard product. Similarly in Section \ref{Random Codes}, we found that binary codes have the most representational power of the suitable codes considered. Hence, we shall primarily compare the memory and capacity of the tensor-spherical scheme to the Hadamard-Rademacher scheme. We also briefly consider other continuous coding schemes paired with the Hadamard product, but we shall see that they are too noisy for accurate graph operations. 

In this section, we focus on the edge query: given a graph embedding, we determine whether it contains a specific edge. For each bind-and-sum scheme considered, we analyze how the dimension of the embedding (memory) and number of edges in superposition (capacity) affect the accuracy of correctly retrieving the right edge from a graph embedding.

We denote the edge binding operation as $\psi$, which can be either the Hadamard product or the tensor product. Let us work with the following fixed graph:
\[
G = \psi(v,u) + \sum_{i=1}^{k} \psi(q_i,r_i)
\]
where all the $q,r$'s are distinct from $u,v$. We perform an edge query to detect the presence of the edge $(u,v)$.

\subsection{Hadamard Product and Rademacher Codes}\label{EQ:HR}
Our graph embedding is of the form:
\[
G = u \odot v + \sum_{i=1}^k q_i \odot r_i
\]
We will first analyze detecting both the signal edge $(u,v)$ and a spurious edge $(a,b)$ that does not exist in the graph. Then, in the case of $M$ competing spurious edges, we analyze the probability of the true edge having a higher edge query score than the $M$ competing spurious ones. Throughout this section, we denote the dimension of the vertex code as $d$, so $u,v,q,r \in \R^d$.

The edge query procedure for Hadamard-Rademacher embeddings proceeds in two steps. First, for a candidate edge $a \odot b$ we unbind it from the graph: $(a \odot b) \odot G$. Second, we sum the result of this unbinding to get a score. The intuition is that if $(a,b)$ exists in the graph, unbinding by it will result in a vector of ones within the superposition. When we sum the result, this vector of ones will return $d$ (the dimension of the vertex code) plus other noise terms. For edge query scores of true and spurious edges, we have the following characterization results.

\begin{theorem} [EQHR: True Query]\label{Thm: EQHRTrue}
Let $G$ be a graph with signal edge $(u,v)$ and $k$ other edges, whose vertices are all distinct:
\[
G = u \odot v + \sum_{i=1}^k q_i \odot r_i
\]
Let $T$ denote the edge query score of signal edge $(u,v)$. We have the following:
\begin{enumerate}
    \item $T$ has the decomposition:
    \[
    T = d + \epsilon
    \]
    where $\frac{\epsilon + kd}{2}$ follows a Binomial$(kd,\frac{1}{2})$ distribution.
    \item $ET = d$
    \item $Var(T) = kd$
    \item $|\epsilon| \propto O(\sqrt{kd})$ with high probabiltiy
\end{enumerate}
\begin{proof}
To compute $T$, we first unbind $(u,v)$ from $G$:
\[
(u \odot v) \odot G = \boldsymbol{1} + \sum_{i=1}^k s_i
\]
where each $s_i = u \odot v \odot q_i \odot r_i$. Unbinding $u \odot v$ by itself results in a vector of ones, denoted $\boldsymbol{1}$. On the other hand, each $s_i$ in the error term is an independent Rademacher vector, as they are products of independent Rademachers. We sum the unbinding result to get the edge query score:
\[
T = sum( (u\odot v) \odot G ) = d + \sum_{i}^k \sum_{j=1}^d s_{ij} = d + \epsilon
\]
where each $s_{ij}$ is a Rademacher random variable. The noise term $\epsilon$ is a sum of $kd$ Rademacher random variables. Using Theorem \ref{Rad Codes} establishes all four claims.
\end{proof}
\end{theorem}

\begin{theorem} [EQHR: Spurious Query]\label{Thm: EQHRFalse}
Let $G$ be a graph with signal edge $(u,v)$ and $k$ other edges, whose vertices are all distinct:
\[
G = u \odot v + \sum_{i=1}^k q_i \odot r_i
\]
Let $F$ denote the edge query score of spurious edge $(a,b)$, whose vertices are distinct from $G$. We have the following:
\begin{enumerate}
    \item $\frac{F + (k+1)d}{2}$ follows a Binomial $((k+1)d,\frac{1}{2})$ distribution.
    \item $EF = 0$
    \item $Var(EF) = (k+1)d$
    \item $|F| \propto O(\sqrt{(k+1)d})$ with high probability.
\end{enumerate}
\begin{proof}
First, we unbind by $a \odot b$:
\[
(a \odot b) \odot G = \sum_{i=1}^{k+1} s'_i
\]
where each $s'_i$ is an independent Rademacher. Summing the result, we see that $F$ is the sum of $(k+1)d$ independent Rademacher random variables. Using Theorem \ref{RadCodeProp} establishes the result.
\end{proof}
\end{theorem}

\begin{theorem}[EQHR: Correct Recovery] \label{Thm:EQHRprob}
Let $G$ be a graph with signal edge $(u,v)$ and $k$ other edges, whose vertices are all distinct:
\[
G = u \odot v + \sum_{i=1}^k q_i \odot r_i
\]
For the signal edge $(u,v)$ and $M$ spurious edges with distinct vertices from $G$, let $A$ denote the event that the true query score  $T$ exceeds all spurious query scores $F_1,\cdots,F_M$. For some constant $C$ we have:
% \[
% (1 - C \sqrt{\frac{2(k+1)}{d}}e^{-\frac{d}{4(k+1)}})^M  \gtrapprox P(A) \geq 1 - M e^{-\frac{d}{(2k+1)}}
% \]
\[
P(A) \geq 1 - M e^{-\frac{C d}{2k+1}}
\]
\begin{proof}
The correct recovery event can be rewritten as:
\[
P(A) = P(T > \max(F_1,\cdots,F_M)) = P(\cap \set{T > F_i}) = 1 - P(\cup \set{T \leq F_i})
\]
By construction, the $F_i$'s have the same distribution. Letting $\epsilon$ denote  the error term:
\begin{align*}
  P(\cup \set{T \leq F_i}) &\leq \sum_{i=1}^M P(T \leq F_i)\\
  &= M P(T \leq F_1)\\
  &= M P(d + \epsilon \leq F_1)\\
  &= M P(F_1 - \epsilon \geq d)
\end{align*}
A difference of Rademacher sums is still a Rademacher sum, so $F_1 - \epsilon$ is a sum of $(k+1)d + kd = (2k+1)d$ Rademachers. Hence, using Bernstein's inequality we bound the term:
\begin{align*}
P(F_1 - \epsilon \geq d) &\leq e^{-\frac{d/(2k+1)}{2(1+(3(2k+1)^2)^{-1})}}\\
&\leq e^{ -\frac{Cd}{2k+1}}
\end{align*}
for some constant $C \geq \frac{1}{4}$. Plugging this in gives:
\[
P(T > \max(F_1,\cdots,F_M)) \geq 1 - M e^{-\frac{Cd}{2k+1}}
\]
% Now, we can compute an approximate upper bound using the above a Gaussian approximation via the CLT. Then, we have:
% \begin{align*}
%     P(T > \max(F_1,\cdots,F_M)) = P(\cap \set{T > F_i}) &= \Pi_i P(T > F_i)\\
%     &= P(T > F_1)^M\\
%     & P(d + \epsilon > F_1)^M\\
%     &= (1- P(F_1 - \epsilon > d))^M\\
%     &\leq (1 - C \sqrt{\frac{2(k+1)}{d}}e^{-\frac{d}{4(k+1)}})^M
% \end{align*}
% where in the last inequality we plugged in the Gaussian upper bound.
\end{proof}
\end{theorem}

\subsection{Hadamard Product and Continuous Codes}
Now, let us suppose the we were working with any continuous code (Gaussian, Cauchy, Uniform). Our edge query would now be unbinding the graph by the reciprocal of the query vertex $u$:
\[
Q(u,G) = u^{-1} \odot (u \odot v + \sum_{i=1}^k q_i \odot s_i) = v + \sum_{i=1}^k u^{-1} \odot q_i \odot s_i
\]
In the noise term, note that we now have a sum of vector whose entries are ratios: $\frac{q_i}{u^{-1}}s_i$. In section \ref{Random Codes}, we saw that the entries will have undefined moments: they follow heavy-tailed distribution. Thus, it is very likely that the noise overwhelms the true answer $v$ regardless of how many edges $k$ are in superposition. This makes such continuous codes infeasible for vertex queries.

\subsection{Tensor Product and Spherical Codes}\label{EQ:TS}
Here, we look at the same error quantities for the tensor-spherical scheme: the edge query scores for true and spurious edges, as well as the probability of correctly identifying the true edge among $M$ spurious edges. As in the previous section, we assume that the vertex codes are $d$-dimensional.
% \subsubsection{Error Norms}
% Now, our edge query is of the form:
% \[
% Q(v, G) = v^T(vu^T + \sum_{i=1}^k q_i r_i^T) = u^T + \sum_{i=1}^k \dprod{v}{q_i}r_i^T = u^T + \sum_{i=1}^k s_i^T
% \]
% We have a corresponding result on the average squared norms and the signal-to-noise ratio.
% \begin{theorem}[Tensor-Spherical Signal-to-Noise] \label{TSEQ:SNR}
% When performing a edge query with a single correct vertex $u$, under the tensor product and spherical codes of dimension $d$ we have:
% \begin{enumerate}
%     \item The squared norm of the signal $E||u||^2$ is $1$.
%     \item The squared norm of the noise $E||\sum_{i=1}^k s_i||^2$ is $\frac{k}{d}$
%     \item The signal-to-noise ratio is $\frac{d}{k}$
% \end{enumerate}
% \begin{proof}
% The first claim holds since spherical codes have norm 1. The squared norm of the nuisance term  $\sum_{i=1}^k \dprod{v}{q_i}s_i^T$ is:
% \[
% E||\sum_{i=1}^k \dprod{v}{q_i}s_i^T||^2 = \sum_i E(\dprod{v}{q_i})^2 + 2\sum_{j \neq k} E\dprod{v}{r_j}\dprod{v}{r_k}\dprod{r_j}{r_k} = \frac{k}{d}
% \]
% Hence, the ratio of the answer-noise average norms is $\frac{d}{k}$

% \end{proof}
% \end{theorem}

% \subsubsection{Statistical Error}
For tensor-spherical embeddings, the graph embedding is:
\[
G = uv^T + \sum_{i=1}^k q_i r_i^T
\]
As in Section \ref{GraphOp}, the edge query for candiate edge $(a,b)$ takes the form:
\[
a^T G b
\]
For the tensor-spherical scheme, we give the following analogous results for the true and spurious edge queries.

\begin{theorem}[EQTS: True Query]\label{Thm:EQTSTrue}
Let $G$ be a graph with signal edge $(u,v)$ and $k$ other edges, whose vertices are all distinct:
\[
G = u v^T + \sum_{i=1}^k q_i r_i^T
\]
Let $T$ denote the edge query score of the signal edge $(u,v)$. We have the following:
\begin{enumerate}
    \item T has the decomposition:
    \[
    T = 1 + \sum_{i=1}^k \dprod{v}{q_i} \dprod{u}{r_i} = 1 + \sum_{i=1}^k \epsilon_i
    \]
    where each $\frac{\epsilon+1}{2}$ is the product of two independent Beta$(\frac{d-1}{2},\frac{d-1}{2})$ rv's.
    \item $ET = 1$
    \item $Var(T) = \frac{k}{d^2}$
    \item $|T-1| \propto O(\frac{\sqrt{k}}{d})$ with high probability.
\end{enumerate}
\begin{proof}
We first compute the edge query score:
\[
T = u^T G v =  u^T (uv^T + \sum_{i=1}^k q_i r_i^T ) v = 1 +  \sum_{i=1}^k \dprod{u}{q_i} \dprod{v}{r_i} = 1 + \sum_{i=1}^k \epsilon_i
\]
Note that by assumption, the dot products $\dprod{u}{q_i}$ and $\dprod{v}{r_i}$ are all independent of each other. Therefore, Theorem \ref{SpherCodeProp} establishes the first claim. Moreover, by the same theorem we have:
\[
E \epsilon_i = 0 \quad ; \quad E \epsilon_i^2 = \frac{1}{d^2}
\]
This establishes the second and third claims. Finally, we see that each $\epsilon_i$ is a random variable absolutely bounded by $1$ with mean 0 and variance $\frac{1}{d^2}$. Therefore, using Bernstein's inequality:
\[
P(\sum_{i=1}^k \epsilon_i \geq C \frac{\sqrt{k}}{d}) \leq \exp(\frac{-C^2 k/d^2}{\frac{k}{d^2} + \frac{C \sqrt{k}}{3d}}) = \exp(\frac{-C^2}{1+C\frac{d}{3\sqrt{k}}}) \approx \exp{-C}
\]
This establishes the fourth claim.
\end{proof}
\end{theorem}
% Now, we again want to recover the answer $u$ by finding which vertex embedding the query output is most similar to, and we will again use the dot product to measure similarity. First, the dot product of the true answer $u$ with the query output:
% \[
% T = u^T (u + \sum_{i=1}^k \dprod{v}{q_i}r_i ) = 1 +  \sum_{i=1}^k \dprod{v}{q_i} \dprod{u}{r_i} = 1 + \epsilon
% \]
% The noise term $\epsilon$ is a sum of $k$ terms of the form $e_i = \dprod{v}{q_i} \dprod{u}{r_i}$, and using independence and the Cauchy-Schwarz inequality:
% \[
% E e_i = 0 \quad ; \quad E e_i^2 = \frac{1}{d^2} \quad ; \quad  E|e_i| \leq \frac{1}{d}
% \]
% Hence, the variance of $\epsilon$ is $\frac{k}{d^2}$, and so for accurate retrieval we see that $k \leq O(d^2)$ or else $\epsilon$ will be of the same magnitude as the signal 1. Similarly, the dot product of a false vertex $t$ (not matching any of the $v_i$'s) is:
\begin{theorem}[EQTS: Spurious Query] \label{Thm:EQTSFalse}
Let $G$ be a graph with signal edge $(u,v)$ and $k$ other edges, whose vertices are all distinct:
\[
G = u v^T + \sum_{i=1}^k q_i r_i^T
\]
Let $F$ denote the edge query score of the spurious edge $(a,b)$ whose vertices are distinct from $G$. We have the following:
\begin{enumerate}
    \item $F$ has the decomposition:
    \[
    F = \sum_{i=1}^{k+1} \epsilon_{i_1} \epsilon_{i_2}
    \]
    where each $\frac{\epsilon+1}{2}$ is an i.i.d Beta$(\frac{d-1}{2},\frac{d-1}{2})$ rv.
    \item $EF = 0$.
    \item $Var(F) = \frac{k+1}{d^2}$.
    \item $|F| \propto O(\frac{\sqrt{k+1}}{d})$ with high probability
\end{enumerate}
\begin{proof}
We compute the edge query:
\[
F = a^T G b = a^T(u v^T + \sum_{i=1}^k q_i r_i^T) b = \dprod{a}{u}\dprod{b}{v} + \sum_{i=1}^k \dprod{a}{q_i} \dprod{b}{r_i}
\]
By construction, each of the summands is the product of independent dot products. Therefore, the result follows using the same argument as in Theorem \ref{Thm:EQTSTrue}.
\end{proof}
\end{theorem}
% Again, we first informally analyze the probability that $F$ exceeds the $1 = ET$. However, since $F$ is sum of random variables with a different distribution, we make one further simplification. By Theorem \ref{SpherCodeProp}, the term $ \dprod{t}{u}$ will be of the order $\frac{1}{\sqrt{d}}$ with high probability. From the sub-Gaussian concentration inequality used in Theorem \ref{DP Dist}:
% \[
% P(|\dprod{t}{u}| > \frac{C}{\sqrt{d}}) \leq \exp{-C^2/2} 
% \]

% Hence, we assume that $|\dprod{t}{u}| = O(\frac{1}{\sqrt{d}})$, and for large $d$ we can make the following simplification:
% \[
% P(F > 1) = P(\epsilon > 1 - O(\frac{1}{\sqrt{d}})) \approx P(\epsilon > 1)
% \]
% Hence, as $\epsilon = \sum^k e_i$ where $e_i$'s are independent with $Ee_i^2 = \frac{1}{d^2}$, then Bernstein's inequality gives:
% \[
% P(F > 1) \approx P(\epsilon > 1) \leq e^{-1/[2(\frac{k}{d^2} + \frac{1}{3})]} \approx e^{-\frac{d^2}{2k}}
% \]
% A similar computation. using the fact that $\epsilon$ is symmetrically distributed, gives an upper bound of $e^{-\frac{d^2}{k}}$ for $P(T < 0)$. Hence, both suggest that the limit of edges we can store in superposition an still have accurate recovery is $O(d^2)$.

% As in the Hadamard-Rademacher case, we have corresponding bound on the probability of accurate recovery for the tensor-spherical scheme.
\begin{theorem} [EQTS: Correct Recovery]\label{Thm:EQTSProb}
Let $G$ be a graph with signal edge $(u,v)$ and $k$ other edges, whose vertices are all distinct:
\[
G = u v^T + \sum_{i=1}^k q_i r_i^T
\]
For the signal edge $(u,v)$ and $M$ spurious edges with distinct vertices from $G$, let $A$ denote the event that the true query score  $T$ exceeds all spurious query scores $F_1,\cdots,F_M$. For some constant $C$ we have:
% \[
% (1 - C \sqrt{\frac{k}{d^2}}e^{-\frac{d^2}{k}})^M  \gtrapprox P(A) \geq 1 - M e^{-\frac{d^2}{k}}
% \]
\[
P(A) \gtrapprox 1 - M e^{-\frac{C d^2}{k}}
\]
\begin{proof}
We express the event of correct recovery as:
\[
P(A) = P(T > \max(F_1,\cdots,F_M)) = 1 - P(\cup \set{T \leq F_i})
\]
By construction, each $F_i$ has the same distribution. :
\begin{align*}
    P(\cup T \leq F_i) &\leq \sum P(T \leq F_i)\\
    &=  M P(T \leq F_1)\\
    &= M P(1 + \epsilon \leq F)\\
    &\leq M P(F - \epsilon \geq 1 )
\end{align*}
Note that $\epsilon$ is the sum of iid random variables symmetrically distributed about 0, so $-\epsilon$ has the same distribution as $\epsilon$. Therefore, the difference $F- \epsilon_1$ is the the sum of $2k$ independent terms, absolutely bounded by 1 with mean 0 and variance $\frac{1}{d^2}$. Hence, using Bernstein's inequality gives:
\[
P(F - \epsilon \geq 1) \leq e^{-\frac{1}{2(2k/d^2+1/3)}} \approx e^{-\frac{Cd^2}{k}}
\]
for some large constant $C$ assuming $k \geq d^2$. Thus,
\[
P(T > \max(F_1,\cdots,F_M)) \gtrapprox 1 - M e^{-\frac{Cd^2}{k}}
\]
\end{proof}
\end{theorem}
\subsection{Impact of Vertex Connectivity}\label{EQ:ConnCorrection}
In the previous theoretical analyses, we assumed that the query edge had distinct vertices from the rest of the edges. This assumption is unrealistic in practical settings, and we will consider the more realistic case of the signal edge sharing vertices with the other edges. We shall see that when the number of shared vertices is low relative to the dimension $d$, the results of the previous analyses still hold.

We again work with a graph with one signal edge and $k$ nuisance edges. However, this time we assume that $L$ of the nuisance edges share one vertex with the signal edge. We shall see in our analysis that the shared vertex being the source $u$ or target $v$ does not matter, so WLOG we assume the shared vertices are only the source vertex $u$:
\[
G = \psi(u,v) + \sum_{i=1}^L \psi(u,p_i) + \sum_{j=1}^{k-L} \psi(q_j,r_j)
\]
where the $p,q,r$'s are distinct from $u,v$. 

\subsubsection{Hadamard Prodcut and Rademacher Codes}
For the Hadamard-Rademacher scheme, none of the previous results will change. Looking at the error norms, our true edge query is:
\[
T = sum(u \odot v \odot G) = sum(\boldsymbol{1} + \sum_{i=1}^L s_i + \sum_{j={L+1}}^{k-L} s_j) = d + \sum_{ij} s_{ij}
\]
The product of independent Rademachers is still Rademacher, so this nice absorbing property makes the error terms exactly the same as in the distinct vertex case. A similar analysis shows that nothing changes for the spurious case, and we see that the results of Section \ref{EQ:HR} hold without modification.

\subsubsection{Tensor Product and Spherical Codes}
For tensor-spherical embeddings, more care needs to be taken. First, we look at the results of our edge query:
\[
 u^T G v = 1 + \sum_{i=1}^L \dprod{p_i}{v} + \sum_{j=1}^{k-L} \dprod{u}{q_j} \dprod{v}{r_j}
\]
Here, notice that $L$ error terms consist of only a single dot product, since those came from the edges that shared the common vertex $u$. The first $L$ terms are mean 0 with variance $\frac{1}{d}$, while the remaining $k-L$ are mean 0 with variance $\frac{1}{d^2}$. Therefore, we can modify the results of Section \ref{EQ:TS} as follows:
\begin{theorem}[EQTS: True Query with Connectivity Correction]\label{Thm:EQTSTrueCorr}
Let $G$ be a graph with signal edge $(u,v)$ and $k$ other edges, among which $L$ edges share a common vertex with $(u,v)$. WLOG, we may assume they all share common source vertex $u$. Then, $G$ takes the form:
\[
G = uv^T + \sum_{i=1}^L u p_i^T + \sum_{j=1}^{k-L} q_j r_j^T
\]
Let $T$ denote the edge query score of the signal edge $(u,v)$. We have the following.
\begin{enumerate}
    \item $ET = 1$
    \item $Var(T) = \frac{L}{d} + \frac{k-L}{d^2} = \sigma^2$
    \item $|T -1| \propto O(\sigma)$ with high probability. For $k >> L$, $T - 1 \propto O(\frac{\sqrt{k}}{d})$
\end{enumerate}
\begin{proof}
As mentioned above, our edge query is:
\[
 u^T G v = 1 + \sum_{i=1}^L \dprod{r_i}{v} + \sum_{j=1}^{k-L} \dprod{u}{q_j} \dprod{v}{r_j} = 1 + \epsilon
\]
From the computation, it is easy to see that whether the edges share $u$ or $v$ does not matter, so WLOG we can assume they all share $u$. Since we are working with spherical codes, by rotational symmetry we see that the first $L$ error terms are independent of each other and the remaining $k-L$ error terms. Hence, the variance of the error sum is the sum of the variances. This, along with each error term being mean 0, establishes the first two claims. Note that each error term is absolutely bounded by $1$ and is symmetrically distributed. Therefore, using Bernstein's inequality:
\begin{align*}
    P(|T - 1| > C\sigma) = P(|\epsilon| > C\sigma) &\leq 2 \exp(\frac{-C^2 \sigma^2}{2(\sigma^2 + \frac{C\sigma}{3})})\\
    &= 2 \exp{\frac{-C^2}{1+\frac{C}{3\sigma}}}
\end{align*}
This establishes the third claim. For large $k$, the second term will dominate and so $T-1 \propto O(\frac{\sqrt{k}}{d})$.
\end{proof}
\end{theorem}
\begin{theorem}[EQTS: Correct Recovery with Connectivity Correction]\label{Thm:EQTSProbCorr}
Let $G$ be a graph with signal edge $(u,v)$ and $k$ other edges, among which $L$ edges share a common vertex with $(u,v)$. WLOG, we may assume they all share common source vertex $u$, so $G$ takes the form:
\[
G = \psi(u,v) + \sum_{i=1}^L \psi(u,p_i) + \sum_{j=1}^{k-L} \psi(q_j,r_j)
\]
For the signal edge $(u,v)$ and $M$ spurious edges that are not in $G$ but may share one vertex with an edge in $G$, let $A$ denote the event that the true query score  $T$ exceeds all spurious query scores $F_1,\cdots,F_M$. For some constant $C$ we have:
% \[
% (1 - C \sqrt{\frac{k}{d^2}}e^{-\frac{d^2}{k}})^M  \gtrapprox P(A) \geq 1 - M e^{-\frac{d^2}{k}}
% \]
\[
P(A) \geq 1 - M \exp(\frac{-1/2}{\frac{2L+1}{d} + \frac{2(k-L)}{d^2} + \frac{1}{3}}) 
\]
For $k >> L$, for some constant $C$ this reduces to:
\[
P(A) \geq 1 - M \exp(\frac{-Cd^2}{k})
\]
\begin{proof}
Note that we assumed that each spurious edge may share one of its vertices with a an edge in $G$. However, since they are not in $G$ they cannot share both vertices with an edge. Moreover, it is easy to see the the more spurious edges that share a vertex in $G$, the larger the spurious query result will be. Therefore, we work with the worst case scenario and assume that each spurious has a common vertex with $G$. WLOG, we may assume that they all share the same source vertex as the signal edge $(u,v)$.

Firstly, from the proof of Theorem \ref{Thm:EQTSProbCorr} we see that the true query can be decomposed as:
\[
T = 1 + \epsilon_1 + \epsilon_2
\]
Similarly, the edge query of the spurious edge $(u,b)$ is:
\[
F = \dprod{v}{b} + \sum_{i=1}^{L} \dprod{p_i}{b_i} + \sum_{j=1}^{k-L} \dprod{q_j}{u} \dprod{r_j}{b} = \epsilon'_1 + \epsilon'_2
\]
Here, $\epsilon_1$ and $\epsilon'_1$ are the sums of i.i.d terms that are absolutely bounded by 1 with mean 0 and variance $\frac{1}{d}$, while $\epsilon_2$ and $\epsilon'_2$ are the sums of i.i.d terms that are absolutely bounded by 1 with mean 0 and variance $\frac{1}{d^2}$. Both of these terms are symmetrically distributed about 0.

We follow the proof of Theorem \ref{Thm:EQTSProbCorr} until we need to bound the term:
\begin{align*}
    P(T \leq F_1) &= F(1 + \epsilon_1 + \epsilon_2 \leq \epsilon'_1 + \epsilon'_2)\\
    P(\epsilon'_1 - \epsilon_1 + \epsilon'_2 - \epsilon_2 \geq 1)
\end{align*}
Using Bernstein's inequality and proceeding with the same steps gives the result:
\begin{align*}
 P(\epsilon'_1 - \epsilon_1 + \epsilon'_2 - \epsilon_2 \geq 1) \leq \exp(\frac{-1/2}{\frac{2L+1}{d} + \frac{2(k-L)}{d^2} + \frac{1}{3}})   
\end{align*}
\end{proof}
\end{theorem}
How do we interpret the results of Theorem \ref{Thm:EQTSProbCorr}? The vertex code dimension $d$ is a hyperparameter and independent of a graph's vertex connectivity $L$, so we focus on the trade-off between $L$ and the number of edges $k$. If we assume that $L$ is independent of $k$ or weakly scales with $k$, then in the limit we see that the second term in the denominator $\frac{2(k-L)}{d}$ dominates and we recover the approximate scaling of $k = O(d^2)$. Theorem \ref{Thm:EQTSTrueCorr} also confirms the scaling $k = O(d^2)$ for large $k$. This large $k$ assumption holds for a variety of applications, such as large knowledge graphs or social networks. In a social network, the number of friends a single person has is unlikely to scale much, if at all, as more people are added to the network past a certain point; one can have only so many friends. However, in the case where $L$ is proportional to $k$, the first term dominates and the bound implies that $k = O(d)$. This situation applies to densely connected graphs, where the node connectivity is proportional to the total number of nodes.

Hence, this suggests that tensor-spherical embeddings are well suited for graphs where the node connectivity weakly scales with or is independent of the number of edges. In particular, this suggests that tensor-spherical embeddings are well suited for large sparse graphs, and we shall cover this connection in Section \ref{AdjacencyMat}. 

\subsection{Memory and Capacity Scaling}
Recall that the vertex code has dimension $d$. Hadamard-Rademacher embeddings have dimension $d$, and previous analysis suggests that we can store at most $O(d)$ edges in superposition without seriously affecting the accurate retrieval of the answer. On the other hand, using the tensor product with spherical codes, the graph embedding space is $d^2$, and the previous analysis shows that we can store at most $O(d^2)$ edges without affecting accuracy. Hence, they have the same memory-capacity scaling with respect to $d$.

This memory-capacity scaling was for the idealized case where all vertices are assumed to be distinct, which might be unrealistic for practical applications. From the connectivity correction of Section \ref{EQ:ConnCorrection}, the memory-capacity scaling of tensor-spherical embedding holds when the number of edges $k$ is much larger than the vertex connectivity $L$. On the other hand, the memory-capacity scaling of Hadamard-Rademacher embeddings are unaffected. Therefore, when dealing with graphs whose number of edges is much larger than their vertex connectivity, the two schemes still have the same memory-capacity scaling.

\section{Binding Comparison: Edge Composition}\label{EdgeComp}
In this section, we focus on the accuracy of edge composition. We are given a graph that contains just two composable edges $(u,v)$ and $(v,w)$, with all other edges having disjoint vertices. Edge composition should return only the valid edge $(u,w)$, and we perform a edge query checking if the embedding contains $(u,w)$. The analysis of this section follows the same procedure as the previous section.

We work with the following fixed graph, which has the two composable edges $(u,v)$ and $(v,w)$ along with $k-2$ nuisance edges :
\[
G = \psi(u,v) + \psi(v,w) + \sum_{i=1}^{k} \psi(q_i,r_i)
\]
where all the $q,r$'s are distinct from $u,v,w$ and $\psi$ denotes the binding operation. We will look at edge composition, checking via an edge query for the correct composition of the composed edge $(u,w)$.

\subsection{Hadamard Product and Rademacher Codes}\label{EC:HS}
% \subsubsection{Error Norms}
We want to do edge composition with $G$, which in this case represents just the binding of $G$ with itself:
\begin{align*}
G \odot G &= (u \odot v + v \odot w + \sum_{i=1}^{k-2} q_i \odot r_i) \odot (u \odot v + v \odot w + \sum_{j=1}^{k-2} q_j \odot r_j)\\
&= u \odot w + \boldsymbol{1} + \sum_{i=1}^{k^2 -2} s_i 
\end{align*}
where each $s_i$ is an independent Rademacher vector. After performing edge composition, we have a superposition of the single composed edge $u \odot w$ along with some nuisance terms. We then perform an edge query to detect the presence of a composed edge, and we have the following results for the true and spurious queries.

\begin{theorem}[ECHR: True Query]\label{Thm:ECHRTrue}
Let $G$ be a graph with the composable edges $(u,v)$ and $(v,w)$ along with $k-2$ other edges, whose vertices are all distinct:
\[
G = u \odot v + v \odot w + \sum_{i=1}^{k-2} q_i \odot r_i
\]
Let $T$ denote the edge query score of the composed edge $(u,w)$ after performing the edge composition procedure $G \odot G$. We have the following:
\begin{enumerate}
    \item T has the decomposition:
    \[
    T = d + \epsilon
    \]
    where $\frac{\epsilon + (k^2 -1)d}{2}$ follows a Binomial$(\frac{(k^2-1)d}{\frac{1}{2}})$ distribution.
    \item $ET = d$
    \item $Var(T) = (k^2+1)d$
    \item $|\epsilon| \propto O(k \sqrt{d})$ with high probability.
\end{enumerate}
\begin{proof}
From the above discussion, performing edge composition results in:
\[
G \odot G = u \odot w + \boldsymbol{1} + \sum_{i=1}^{k^2 -2} s_i
\]
where each $s_i$ is an independent Rademacher vector. We then query for true edge $(u,w)$ by first unbinding by it:
\[
(u \odot w) \odot (G \odot G) = \boldsymbol{1} + u \odot w + \sum_{i=1}^{k^2 -2} s'_i
\]
where again each $s'_i$ is an independent Rademacher vector. We then sum this result to get our edge query, which results in:
\[
d + \sum_{i}^{k^2-1} \sum_{j=1}^d s'_{ij}
\]
where each $s'_{ij}$ is a Rademacher random variable. Using the same techniques as in Theorem \ref{Thm: EQHRTrue} gives the result. Note that for the fourth claim, $\epsilon \propto O(\sqrt{(k^2-1)d}) = O(k\sqrt{d})$.
\end{proof}
\end{theorem}
\begin{theorem}[ECHR: Spurious Query]\label{Thm:ECHRFalse}
Let $G$ be a graph with the composable edges $(u,v)$ and $(v,w)$ along with $k-2$ other edges, whose vertices are all distinct:
\[
G = u \odot v + v \odot w + \sum_{i=1}^{k-2} q_i \odot r_i
\]
Let $F$ denote the edge query score of the spurious edge $(a,b)$ after performing the edge composition procedure $G \odot G$. We have the following:
\begin{enumerate}
    \item $\frac{F + k^2d}{2}$ follows a Binomial$(\frac{k^2d}{\frac{1}{2}})$ distribution.
    \item $ET = 0$
    \item $Var(T) = k^2d$
    \item $|F| \propto O(k \sqrt{d})$ with high probability.
\end{enumerate}
\begin{proof}
The proof is analogous to that of Theorem \ref{Thm:ECHRTrue}.
\end{proof}
\end{theorem}

\begin{theorem}[ECHR: Correct Recovery]\label{Thm:ECHRProb}
Let $G$ be a graph with the composable edges $(u,v)$ and $(v,w)$ along with $k-2$ other edges, whose vertices are all distinct:
\[
G = u \odot v + v \odot w + \sum_{i=1}^{k-2} q_i \odot r_i
\]
For the signal edge $(u,w)$ and $M$ spurious edges with distinct vertices from $G$, let $A$ denote the event that the true query score  $T$ exceeds all spurious query scores $F_1,\cdots,F_M$. For some constant $C$ we have:
% \[
% (1 - C \sqrt{\frac{4(k-1)^2 - 2}{d}}e^{-\frac{d}{4(k-1)^2 - 2}})^M \geq P(A) \geq 1 - M e^{-\frac{d}{4(k-1)^2}}
% \]
\[ 
P(A) \geq 1 - M e^{-\frac{C d}{2k^2 - 1}}
\]
\begin{proof}
We proceed as in Theorem \ref{Thm:EQHRprob} until we bound the probability:
\[
P(T \leq F_1) = P(F_1 - \epsilon \geq d)
\]
The difference $F_1 - \epsilon$ is the sum of $k^2 + k^2-1 = 2k^2-1$ indpendent Rademachers. Hence, using the standard Rademacher concentration inequality gives the result.
\end{proof}
\end{theorem}

\subsection{Tensor Product and Spherical Codes}\label{EC:TS}
For tensor-spherical embeddings, our graph embedding is:
\[
G = uv^T + vw^T + \sum_{i=1}^{k-2} q_i r_i^T
\]
and our edge composition procedure is the matrix power:
\[
G^2 = (uv^T + vw^T + \sum_{i=1}^{k-2} q_i r_i^T)(uv^T + vw^T + \sum_{i=1}^{k-2} q_i r_i^T)
\]
We then perform edge query to detect the presence of a composed edge, and we have the following analogous results for the true and spurious queries.
\begin{theorem}[ECTS: True Query]\label{Thm:ECTSTrue}
Let $G$ be a graph with the composable edges $(u,v)$ and $(v,w)$ along with $k-2$ other edges, whose vertices are all distinct:
\[
G = uv^T + vw^T + \sum_{i=1}^{k-2} q_i r_i^T
\]
Let $T$ denote the edge query score of the composed edge $(u,w)$ after performing the edge composition procedure $G^2$. We have the following:
\begin{enumerate}
    \item $T$ has the decomposition:
    \[
    T = 1 + \sum_{i=1}^{2k-2} \epsilon_{i_1} \epsilon_{i_2} + \sum_{j=1}^{k^2-2k+1} \epsilon_{j_1} \epsilon_{j_2} \epsilon_{j_3}
    \]
    where each $\frac{\epsilon+1}{2}$ is an i.i.d Beta($\frac{d-1}{2},\frac{d-1}{2}$) rv.
    \item $ET =1$
    \item $Var(T) = \frac{2k-2}{d^2} + \frac{k^2 - 2k +1}{d^3}$
    \item $|T-1| \propto O(\sqrt{\frac{2k-2}{d^2} + \frac{k^2 - 2k +1}{d^3}})$ with high probability. For large $k$, $T \propto O(k d^{\frac{-2}{3}})$.
    
\end{enumerate}
\begin{proof}
Consider the edge composition procedure, represent by the matrix power:
\[
G^2 = (uv^T + vw^T + \sum_{i=1}^{k-2} q_i r_i^T) G
\]
Distributing each term on the left into the right, $uv^T$ will return the composed edge $uw^T$ and $k-1$ terms of the form $\dprod{u}{s} ut^T$. Similarly, $vw^T$ will distribute and return the term $\dprod{w}{v} vw^T$ and $k-1$ terms of the form $\dprod{w}{s}v t^T$, and each $q_ir_i^T$ distributes and returns one term $\dprod{q_i}{v} q_i w^T$ and $k-1$ terms of the form $\dprod{q_i}{s}q_i t^T$. Hence, edge composition will generate the true composed edge $uw^T$, $2k-2$ edges that share one vertex with $(u,w)$, and $k^2-2k+1$ remaining edges that are disjoint from $(u,w)$. After querying for the edge $(u,w)$, we have:
\[
T = 1 + \sum_{i=1}^{2k-2} d_{i_1} d_{i_2} + \sum_{j=1}^{k^2-2k+1} d_{j_1} d_{j_2} d_{j_3}
\]
where each $d$ is a dot product between two different spherical codes. Note that the $d$'s are i.i.d, and so by Theorem \ref{SpherCodeProp} we have:
\[
ET = 0 \qquad ; \qquad Var(T) = \frac{2k-2}{d^2} + \frac{k^2-2k+1}{d^3}
\]
This establishes the first three claims.

For the fourth, note that each error is absolutely bounded by 1. Bernstein's inequality shows that $T \propto O(\sqrt{\frac{2k-2}{d^2} + \frac{k^2-2k+1}{d^3}})$ with high probability. Since we are interested in large $k$ (many edges), the second term dominates and we have $T \propto O(\sqrt{\frac{k^2-2k+1}{d^3}}) = O(kd^{\frac{-2}{3}})$.
\end{proof}
\end{theorem}
\begin{theorem}[ECTS: Spurious Query]\label{Thm:ECTSFalse}
Let $G$ be a graph with the composable edges $(u,v)$ and $(v,w)$ along with $k-2$ other edges, whose vertices are all distinct:
\[
G = uv^T + vw^T + \sum_{i=1}^{k-2} q_i r_i^T
\]
Let $F$ denote the edge query score of the spurious edge $(a,b)$ whose vertices are distinct from $G$, after performing the edge composition procedure $G^2$. We have the following:
\begin{enumerate}
    \item $F$ has the decomposition:
    \[
    F = \sum_{i=1}^{k^2-1} \epsilon_{i_1} \epsilon_{i_2} \epsilon_{i_3} +  \epsilon_1 \epsilon_2
    \]
    where each $\epsilon$ is an i.i.d Beta($\frac{d-1}{2},\frac{d-1}{2}$) rv.
    \item $EF =0$
    \item $Var(F) = \frac{k^2 -1}{d^3} + \frac{1}{d^2}$
    \item $|F| \propto O(\sqrt{\frac{k^2 -1}{d^3} + \frac{1}{d^2}})$ with high probability. For large $k$, $F \propto O(kd^{\frac{-2}{3}})$.
\end{enumerate}
\begin{proof}
Since we query by the spurious edge $(a,b)$ whose vertices are disjoint, we do not need  to keep track of edges that share vertices with $(a,b)$. Hence, after edge composition we have will have $uw^T$ and $k^2 -1$ terms of the form $\dprod{s}{t}qr^T$. Distributing the edge query to each of these terms, we will have one term that is the product of two dot products and $k^2-1$ terms that are the product of three. Note that each dot product is i.i.d, and Theorem \ref{SpherCodeProp} establishes the first three claims. Bernstein's inequality shows that $F \propto O(\sqrt{\frac{k^2 -1}{d^3} + \frac{1}{d^2}})$ with high probability, and when $k$ is large we see that the first term dominates. This establishes the fourth claim.
\end{proof}  
\end{theorem}

\begin{theorem}[ECTS: Correct Recovery]\label{Thm:ECTSProb}
Let $G$ be a graph with the composable edges $(u,v)$ and $(v,w)$ along with $k-2$ other edges, whose vertices are all distinct:
\[
G = uv^T + vw^T + \sum_{i=1}^{k-2} q_i r_i^T
\]
For the signal edge $(u,w)$ and $M$ spurious edges with distinct vertices from $G$, let $A$ denote the event that the true query score  $T$ exceeds all spurious query scores $F_1,\cdots,F_M$. Let $\sigma^2 = \frac{2k-1}{d^2} + \frac{2k^2-2k}{d^3}$. For some constant $C$ we have:
\[
P(A) \gtrapprox 1 - M \exp(\frac{-C}{\sigma^2})
\]
For large $k$, this reduces to:
\[
P(A) \gtrapprox 1- M\exp(\frac{-Cd^3}{k^2})
\]
\begin{proof}
We proceed as in the analogous theorems until we need to bound:
\[
P(T \geq F_1) = P(1 + \epsilon \geq F_1) = P(\epsilon - F_1 \leq 1)
\]
From Theorems $\ref{Thm:ECTSTrue}$ and \ref{Thm:ECTSFalse}, we see $\epsilon - F_1$ is the sum of $2k-1$ terms that are the products of two dot products and $2k^2 -2k$ terms that are the products of three dot products. Each term is absolutely bounded by 1, and let $\sigma^2 = \frac{2k-1}{d^2} + \frac{2k^2-2k}{d^3}$. By Bernstein's inequality:
\begin{align*}
    P(T \geq F_1) = P(\epsilon - F_1 \leq 1) \leq \exp(\frac{-\frac{1}{2}}{\sigma^2 + \frac{1}{3}}) \approx \exp(\frac{-1}{2\sigma^2})
\end{align*}
as long as $\sigma^2$ is large relative to $\frac{1}{3}$. For large $k$, we see that the $\frac{2k^2-2k}{d^3}$ will dominate, and so:
\[
P(T \geq F_1) \gtrapprox \exp(\frac{-Cd^3}{k^2})
\]
\end{proof}
\end{theorem}

\subsection{Impact of Vertex  Connectivity}\label{EC:ConnCorr}
We analyze the impact of vertex connectivity on the previous analyses, where assumed the nuisance edges had distinct vertices. Here, we consider the case where the nuisance edge might share a vertex with either of the composable edges. We shall see that when vertex connectivity is low relative to the dimension $d$, the results of the previous analyses still hold.

We again work with a graph with composable edges $(u,v)$ and $(v,w)$ along with $k-2$ nuisance edges. This time, we assume that $L$ of the nuisances edges share one vertex with either $(u,v)$ and $(v,w)$. We shall see in our analyses that we need to further partition the $L$ edges that share vertices into two case: the nuisance edges that have common end vertices $u$ and $w$, and the nuisance edges that have common vertex $v$ and are composable with $(u,v)$ or $(v,w)$. We shall show that within each of the two above cases, it doesn't matter whether the shared vertex is a source or target vertex. Therefore, we split the $L$ edges into these two cases, and $G$ takes the form:
\[
G = \psi(u,v) + \psi(v,w) + \sum_{i=1}^{L_1} \psi(v,p_i) + \sum_{j=1}^{L_2} \psi(u,p_j) + \sum_{l=1}^{k-L-2} \psi(q_l,r_l)
\]
where $L = L_1 + L_2$.

\subsubsection{Hadamard Product and Rademacher Codes}
As in Section \ref{EQ:ConnCorr}, all the results in the previous section hold without modification because although the nuisance edges might share one common vertex, the presence of at least one distinct vertex makes the bound edge an independent Rademacher vector.

\subsubsection{Tensor Product and Spherical Codes}
Our graph embedding is of the form:
\[
G = uv^T + vw^T + \sum_{i=1}^{L_1} vp_i^T + \sum_{j=1}^{L_2} u p_j^T + \sum_{l=1}^{k-L -2} q_l r_l^T
\]
where $L_1 + L_2 = L$. Correcting for vertex connectivity, we have the following modified results.
\begin{theorem}[ECTS: True Query with Connectivity Correction] \label{Thm:ECTSTrueCorr}
Let $G$ be a graph with the composable edges $(u,v)$ and $(v,w)$ along with $k-2$ nuisance edges. Of the nuisance edges, we assume assume $L_1$ have common vertex $u$ or $w$ and $L_2$ have common vertex $v$, with the other vertex being distinct WLOG, we may assume $G$ takes the form:
\[
G = uv^T + vw^T + \sum_{i=1}^{L_1} vp_i^T + \sum_{j=1}^{L_2} u p_j^T + \sum_{l=1}^{k^2-L -2} q_l r_l^T
\]
Let $T$ denote the edge query score of the composed edge $(u,w)$. Let $L = L_1 + L_2$, $N_1 = 2k + L_2(k-1) - L -2$, and $L_2 = k^2 - 2k - L_2(k-1) - 1$. We have the following:
\begin{enumerate}
    \item $T$ has the decomposition:
    \[
    T = 1 + \sum_{i=1}^L \epsilon_{i}  +   \sum_{j=1}^{N_1} \epsilon_{j_1} \epsilon_{j_2}  + \sum_{l=1}^{N_2} \epsilon_{l_1} \epsilon_{l_2}\epsilon_{l_3}
    \]
    where each $\frac{\epsilon + 1}{2}$ is an i.i.d Beta$(\frac{d-1}{2},\frac{d-1}{2})$ rv.
    \item $ET = 1$
    \item $Var(T) = \frac{L}{d} + \frac{N_1}{d^2} + \frac{N_3}{d^3} = \sigma^2$
    \item $|T - 1| \propto O(\sigma)$ with high probability. For $k >> L$, $T-1 \propto O(kd^{\frac{-3}{2}})$
\end{enumerate}
\begin{proof}
First, we perform edge composition:
\[
(uv^T + vw^T + \sum_{i=1}^{L_1} vp_i^T + \sum_{j=1}^{L_2} u p_j^T + \sum_{l=1}^{k^2-L -2} q_l r_l^T ) G
\]
We distribute and group terms by both the number of erroneous dot products and whether they have source vertex $u$ or target vertex $v$. Let $q,s$ denote generic distinct vertices, and $\epsilon$ the dot product of two different spherical codes. Performing edge composition yields the following number of terms:
\begin{enumerate}
    \item $uw^T$: 1
    \item $ u s$: $L_1$
    \item $\epsilon (uw^T)$: $L_2$
    \item $\epsilon (u s^T)$ or $\epsilon (q w^T)$: $2k + L_2(k-1) - L -2 = N_1$
    \item $\epsilon (q s^T):$ $k^2 - 2k - L_2(k-1) - 1 = N_2$
\end{enumerate}
After querying for edge $(u,w)$, the true query can be expressed as:
\[
T = 1 + \sum_{i=1}^L \epsilon_{i}  +   \sum_{j=1}^{N_1} \epsilon_{j_1} \epsilon_{j_2}  + \sum_{l=1}^{N_2} \epsilon_{l_1} \epsilon_{l_2}\epsilon_{l_3}
\]
where each $\frac{\epsilon + 1}{2}$ is an i.i.d Beta$(\frac{d-1}{2},\frac{d-1}{2})$ rv. Together with Theorem \ref{SpherCodeProp}, this establishes the first three claims. Noting that every term is absolutely bounded by 1 and Bernstein's inequality gives the fourth. Note that $L$, $N_1$, and $N_3$ scale with $k$ to the $0^{th}, 1^{st}$, and $2^{nd}$ degree respectively. Therefore, for large $k >> L$ the term $\frac{N_3}{d^3} \approx \frac{k^2}{d^3}$ will dominate and so $T-1 \propto O(kd^{\frac{-3}{2}})$.

Finally, we show that we may assume $G$ had the assumed form WLOG. First, consider the $L_1$ edges that had common vertex $v$: $v p_i^T$. These edges are special because they can be validly composed with $uv^T$. If they take the alternate form $q_i v^T$, these can be validly composed with $vw^T$. Either way, they result in an edge that is scaled by an $\epsilon$ and has either $u$ as a source vertex or $w$ as a target vertex. After querying by $(u,w)$, both edges will result in the same term $\epsilon_1 \epsilon_2$. Secondly, let us consider the $L_2$ edges that had common source vertex $u$: $u p_j^T$. These are special because they each give rise to $k-1$ terms of the form $\epsilon u s^T$ and one term of the form $\epsilon uw^T$. If we consider the alternate form $q_j w^T$  that have common target vertex $w$, they will each give rise to $k-1$ terms of the form  $\epsilon q w^T$ and one term $\epsilon uw^T$ (easy to see if you distribute from the right). Hence, after the edge query both will result in the terms. Thus, we may indeed assume WLOG that $G$ takes the above form.
\end{proof}
\end{theorem}

\begin{theorem}
Let $G$ be the same as in Theorem \ref{Thm:ECTSTrueCorr}. In particular, recall that $L = L_1 + L_2$, $N_1 = 2k + L_2(k-1) - L -2$, and $L_2 = k^2 - 2k - L_2(k-1) - 1$. Then, $G$ takes the form:
\[
G = uv^T + vw^T + \sum_{i=1}^{L_1} vp_i^T + \sum_{j=1}^{L_2} u p_j^T + \sum_{l=1}^{k^2-L -2} q_l r_l^T
\]
For the signal edge $(u,w)$ and $M$ spurious edges with distinct vertices from $G$, let $A$ denote the event that the true query score  $T$ exceeds all spurious query scores $F_1,\cdots,F_M$. For some constant $C$ we have:
\[
P(A) \geq 1 -M \exp(\frac{-1/2}{\frac{L}{d} + \frac{N_1+L_1+1}{d^2} + \frac{N_2 + k^2 - L_1 -1}{d^3} + \frac{1}{3}})
\]
For $k >> L$, this reduces to:
\[
P(A) \geq 1 - M \exp(\frac{-Cd^3}{k^2})
\]
\begin{proof}
We proceed as usual until we need to bound the event:
\[
P(T \geq F_1)
\]
We first decompose $F_1$, which is the result of querying the edge composition $G^2$ by the spurious edge $(a,b)$. From the list of terms provided in the proof of Theorem \ref{Thm:ECTSTrueCorr}, since $a,b$ are assumed to be disjoint vertices we have:
\[
F = \sum_{i=1}^{L_1 + 1} \epsilon_{i_1} \epsilon_{i_2} + \sum_{j=1}^{k^2 - L_1 + 1} \epsilon_{j_1} \epsilon_{j_2}\epsilon_{j_3}
\]
Hence, combining terms with the decomposition of $T$ from Theorem \ref{Thm:ECTSTrueCorr}:
\[
P(T \geq F_1) = P( \ \sum_{i=1}^{L} \epsilon + \sum_{j=1}^{N_1 + L_1 + 1} \epsilon_{j_1} \epsilon_{j_2} + \sum_{l=1}^{N_2 + k^2 - L_1 -1} \epsilon_{l_1} \epsilon_{l_2}\epsilon_{l_3} \geq 1 \ )
\]
Hence, since each term is absolutely bounded by 1 and independent, by Bernstein's inequality we have:
\[
P(T \geq F_1) \leq \exp(\frac{-1/2}{\frac{L}{d} + \frac{N_1+L_1+1}{d^2} + \frac{N_2 + k^2 - L_1 -1}{d^3} + \frac{1}{3}})
\]
Again, note that $N_1$ and $N_2$ scale linearly and quadratically with $k$ respectively. Hence, for large $k >> L$, the term $\frac{N_2 + k^2 - L_1 -1}{d^3}$ dominates in the denominator. Since $N_2 + k^2 - L_1 -1 = O(k^2)$ for large $k$, the bound reduces to $\exp(\frac{-Cd^3}{k^2})$ 
\end{proof}
\end{theorem}
We see that when the number of edges $k$ is much larger than the vertex connectivity $L$, we recover the results of Section \ref{EQ:TS}. From our discussion in Section \ref{EQ:ConnCorr}, this is true for a variety of applications, such as knowledge and social graphs. One prominent example where this assumption fails is densely connected graphs. However, for applications the local vertex connectivity is low relative to the number of edges, tensor-spherical embeddings are just as efficient as Hadamard-Rademacher embeddings while boasting superior expressivity.

\subsection{Memory and Capacity Scaling}\label{EC:MemCapScal}
The graph embedding dimension under Hadamard product with Rademacher codes is $d$, and the above analysis gives a limit of $\sqrt{d}$ edges that can be stored in superposition. On the other hand, the tensor product with spherical codes has dimension $d^2$ and it can store at most $d^{\frac{3}{2}}$ edges in superposition. In both cases, memory-capacity scaling with respect to $d$ is $\sqrt{d}$. Again, the Hadamard product with Rademacher codes offers no concrete memory advantages over the tensor product since any savings memory are offset by proportional penalty capacity.

Like the memory-capacity discussion for the edge query, the previous scaling assumed the idealized case of the edges all having distinct vertices. However,  the analyses of Section \ref{EC:ConnCorr}, accounting for the case of vertex overlap in the composable edges, show that as long as the number of edges $k$ is much larger than the vertex connectivity $L$, the memory-capacity scaling still holds.

However, note that while they have the same memory-capacity scaling with respect to the vertex code dimension $d$, the two embeddings start to diverge in their memory-capacity ratios in edge composition. Tensor-spherical embeddings of size $d^2$ can store $O(d^{\frac{3}{2}})$ edges accurately. On the other hand, if we consider Hadamard-Rademacher embeddings of the same dimension $d^2$, our analysis shows that they can store $O(\sqrt{d^2}) = O(d)$ edges. In terms of capacity per parameter used, tensor-embeddings are superior to Hadmard-Rademacher embeddings in edge composition. A quick review of the results of Section \ref{EdgeQuer} shows that this is not the case for the edge query, where they have the same memory-capacity ratio. Why do they differ in the edge composition? We discuss this phenomenon in a general setting in the next section.

\section{Binding Comparison Summary}
From the results of Sections \ref{EdgeQuer} and \ref{EdgeComp}, we see that tensor-spherical and Hadamard-Rademacher embeddings have the same memory-capacity scaling with respect to the vertex dimension $d$. We actually do not save any memory by compression via the Hadamard product, since its reduced memory requirement is offset by a proportional penalty on its capacity. In fact, in Section \ref{EC:MemCapScal} we see the Hadamard-Rademacher embeddings effecitvely store less edges per parameter used compared to tensor-spherical embeddings. Since the circular correlation and convolution are special cases of the Hadamard, this immediately suggests they also suffer the same defect. We shall discuss the memory-capacity scaling and ratios in a general setting, as well as its implications for other binding operations.

\subsection{General Memory-Capacity Scaling}\label{GenMemCap}
For both embedding schemes, the edge query is a first-order operation since it can be expressed as a linear operation involving one graph embedding. Similarly, edge composition is a second-order operation since it can be expressed as a multilinear operation involving two graph embeddings. We sketch a general memory-capacity ratio for a $n$-order operation.

For the Hadamard-Rademacher scheme, a $n$-order operation on a graph with $k$ edges will create $k^n$ nuisance terms, giving an error term of with variance $k^n d$. As we saw in previous sections, we need to scale $k$ such that the error term's variance is at most $d^2$. This suggests that at most $d^\frac{1}{n}$ edges can be stored in superposition, and the general memory-capacity scaling for the Hadamard-Rademacher scheme is: 
\[
\frac{d}{d^{1/n}} = d^{1-\frac{1}{n}}= d^{(n-1)/n}
\]

For the tensor-spherical scheme a $n$-order operation will also create $k^n$ edges, but each will be weighted by a random coefficient with mean 0 and variance $d^{-(n+1)}$. During edge recovery, we will then have $k^n$ error terms with mean 0 and variance $d^{-(n+1)}$, creating a total error term with variance approximately $k^nd^{-(n+1)}$. From the analysis of previous sections, $k$ needs to scale such that the error variance is at most 1. This suggests we can store at most $d^{(n+1)/n}$ edges, and the memory-capacity scaling is:
\[
\frac{d^2}{d^{(n+1)/n}} = d^{2-(n+1)/n} = d^{(n-1)/n} 
\]
Therefore, Hadamard-Rademacher embeddings have the same scaling as the tensor-spherical embeddings, and they offer no real memory advantage. The connectivity correction results of Sections \ref{EQ:ConnCorrection} and \ref{EC:ConnCorr} show that this still holds in the more practical case where the signal edge shares vertices with the nuisances edges, as long as the vertex connectivity is small relative to the number of edges.

\subsection{Memory and Capacity: Scaling vs. Ratio} \label{GenMCScalandRat}
We draw an important distinction between memory-capacity scaling and memory-capacity ratios. Sections \ref{EdgeQuer} and \ref{EdgeComp} showed both embeddings have the same memory-capacity scaling with respect to the vertex dimension $d$. However, as we saw in Section \ref{EC:MemCapScal} this does not imply that they have the same memory-capacity ratio. For edge composition, Hadamard-Rademacher embeddings of size $d^2$ can store $O(d)$ edges, which is strictly less the $O(d^{\frac{3}{2}})$ edges tensor-spherical embeddings can store using the same number of parameters. While their memory-capacity scaling is the same, tensor-spherical embeddings in fact have a superior memory-capacity ratio in the edge composition and can store more edges per parameter used.

Now consider the general memory-capacity scaling in the previous section. If we set the Hadamard-Rademacher embeddings to have the same dimension $d^2$ as tensor-spherical embeddings, then they can store $O(d^{\frac{2}{n}})$ edges for a $n$-order operation. On the other hand, tensor-spherical embedding can store $O(d^{\frac{(n+1)}{n}})$ edges in superposition, which is strictly better. We see that for graph operations of order $n \geq 2$, tensor-spherical embeddings have a more efficient capacity scaling per parameter used. These differing memory-capacity ratios stem from properties of the binding operation, and we discuss this in the next section.

\subsection{Tensor Product and Approximate Orthonormality}
In our analyses of tensor-spherical embeddings, the approximate orthonormality of spherical codes  played a large part in reducing the contribution of the error terms. This property is not unique to spherical codes, and we can achieve similar results with other pseudo-orthonormal codes. In this section, we will study the tensor and Hadamard product paired with normalized Rademacher codes (NR codes) to demonstrate the different properties of each binding operation. In particular, we show that the tensor product leverages orthornormality to control the buildup of error terms, a property that the Hadamard product lacks. This explains the results of the previous section, where we saw that tensor-spherical embeddings have a higher memory-capacity ratio for higher order graph operations.

We first establish some preliminary results on NR codes.
\begin{theorem} \label{RademacherResults}
    Let $r_1,\cdots, r_k$ be $d$-dimensional Rademacher vectors, whose entries are iid Rademacher random variables. Let $x_1,\cdots,x_k$ be their normalized versions (ie. $x_i = \frac{1}{\sqrt{d}}r_i$) and let $Z = \dprod{x_i}{x_j}$ denote their dot product for $i \neq j$. The following properties hold:
    \begin{enumerate}
        \item $\frac{d(Z + 1)}{2} \sim Binonmial(d,\frac{1}{2})$.
        \item $E(Z) = 0$ and $Var(Z)=\frac{1}{d}$
        \item $P(\max_{i \neq j}\limits |\dprod{x_i}{x_j}| > \epsilon) \geq 1 - 2\binom{k}{2}e^{-\frac{d}{4}\epsilon^2 }$
    \end{enumerate}
\begin{proof}
    % Since the product of two Rademacher random variables, $Z$ is the sum of $d$ Rademacher random variables scaled by a factor of $\frac{1}{d}$. If $B$ = Bernoulli($\frac{1}{2}$), then $2B-1$ is a Rademacher random variable. Hence, the sum of $d$ Rademacher is $2 Y -d$ where $Y \sim Binomial(d,\frac{1}{2})$, and so $Z = \frac{1}{d}(2Y-d)$. This establishes the first statement. The second statement follows from the first, and the third follows by the same argument as in Theorem \ref{Optimality}.
    The first two statements follow from Theorem \ref{RadCodeProp}. For the third claim, we see $Z$ is average of $d$ iid Rademacher variables. By a standard Rademacher concentration inequality \cite{Boucheron2004} we have:
    \[
    P( Z > \epsilon) \leq e^{-\frac{d\epsilon^2}{2(1+\frac{\epsilon}{3})}} \leq e^{-\frac{d\epsilon^2}{4}}
    \]
    Since $Z$ is symmetrically distributed, the lower-tail bound is the same. Combining the two and taking a union bound gives the result.
\end{proof}
\end{theorem}

Consider  a graph $G$ consisting of two edges $(u,v)$ and $(s,t)$. For tensor-NR embeddings, we check if $G$ contains the edge $(u,v)$ through an edge query:
\[
u^T(uv^T + st^T)v = 1 + (\dprod{u}{s})(\dprod{t}{v})
\]
The error term consists of dot products between two pairs of mismatched vertices, and their pseudo-orthonormality shrinks the error term to be of magnitude $\frac{1}{d}$ with high probability. Mismatched vertices interact destructively, and the tensor product uses the pseudo-orthogonality of NR codes to clean up error terms.

On the other hand, consider the same edge query for Hadamard-NR embeddings:
\[
\frac{1}{d^2}sum[(u \odot v) \odot (u \odot v + s \odot t)]  = \frac{1}{d^2} sum[\boldsymbol{1} + u \odot v \odot s \odot t] = \frac{1}{d^2} (d + \sum_{i=1}^d r_i)
\]
Ignoring the common scaling factor of $\frac{1}{d^2}$, the error term is the sum of $d$ Rademacher random variables. Importantly, even though NR codes are pseudo-orthonormal, there is no destructive interaction in the mismatch terms. Hence, since the number of error terms scales exponentially with the order of a graph operation, this explains the tensor product's superior memory-capacity ratio in Sections \ref{EC:MemCapScal} and \ref{GenMCScalandRat}. Although it is outside the scope of this paper, one can show that in order to effectively leverage orthonormality, a binding operation must result in an embedding as large as the tensor product. In other words, the tensor product is the most compact binding operations that fully preserves all geometric information of its bound vertex codes.

\subsection{Tensor Product and Other Binding Operations}
As mentioned in Section \ref{AlterBindOp}, the circular correlation and convolution are special cases of the Hadamard product, and the conclusions drawn here for the Hadamard product can be extended to them. While our analysis covers only the Hadamard product paired with Rademacher codes, we note two fundamental defects with the Hadamard product and its related embeddings. Firstly, as we saw in Section \ref{AlterBindOp} the Hadamard product sacrifices some espressivity, where it cannot represent directed edges and edge composition at the same time. Secondly, as we saw in this section the Hadamard product has a strictly worse memory-capacity ratio to the tensor product because it does not leverage orthonormality to control error terms.

Conversely, two defects of the tensor product are that its dimension explodes with both the binding order and dimension of the original code. In the context of graph embeddings, the first defect is not a concern because the binding order is always small. The memory-capacity analysis of this and previous sections also showed that the second defect is not a concern relative to other compressed binding operations, and in fact the tensor product can effectively represent more edges per parameter used than these alternative binding operations.

\section{Relationship to Adjacency Matrices}\label{AdjacencyMat}
All of the graph operations covered in Section \ref{GraphOp} are reminiscent of operations one can do with adjacency matrices.  In fact, tensor-spherical embeddings are random analogues of adjacency matrices, and the previous memory-capacity analyses suggest that they may serve as a compressed representation for large sparse matrices. We explore this connection and its implications in this section.

\subsection{Adjacency Matrices and Simple Tensors}
A simple tensor of the tensor product $V \otimes V$ is any tensor that can be written as $v \otimes w$ for $v,w \in V$. Fixing a basis, an adjacency matrix $A$ can be expressed as the sum of the coordinate simple tensors:
\[
A = \sum_{ij: A_{ij} = 1} e_i \otimes e_j = \sum_{ij: A_{ij} = 1} e_i e_j^T
\]
where $e_i$ are the standard coordinate vectors. Note that the coordinate vectors form an orthonormal basis. Let $B = \set{b_1,\cdots,b_d}$ be another orthonormal basis, with the associated change of basis matrix $P$.
Changing bases to $B$, the coordinate form $A_B$ of the original adjacency matrix $A$ in the new basis $B$ satisfies the following equation:
\[
A_B = \sum_{ij: A_{ij}=1} b_i b_j^T = P A P^T
\]
Hence, matrices that can be expressed as sums of simple tensors of orthonormal vectors are equivalent to adjacency matrices, up to an orthogonal change of basis. Using exact orthonormal codes, our graph embedding $G$ of graph $\boldsymbol{G}$ coincides with  its adjacency matrix $A$, after an orthogonal change of basis $P$:
\[
G = P A P^T
\]
This partly explains their rich representational capacity, since they are capable of any graph operation adjacency matrices can do. However, our use of pseudo-orthogonal vectors paves the way forward for drastic compression, especially for large sparse matrices. For a graph with $d$ vertices and $k$ edges, adjacency matrices scales as $O(d^2)$. However, if we are concerned with just first-order graph operations the analysis of Section \ref{EdgeQuer} shows that replacing orthonormal vectors with pseudo-orthonormal vectors results in a graph representation that scales as $O(k)$. This has implications toward a dynamic compressed graph representations, especially for large sparse matrices. We shall explore this connection in a subsequent section.

\subsection{Generalizing Spectral Machinery} \label{SpectralMachine}
Using exact orthonormal codes, tensor-spherical embeddings are equal to adjacency matrices, up to a change of basis. Immediately, they share any basis-independent property of adjacency matrices, including the spectrum and its related machinery. Given a vector space $W$ and a orthgonal basis $B$, we can define the generalized diagonal operator on $W \otimes W$ as the operator that only preserves the self-similar tensors:
\[
Diag_B: W \otimes W \rightarrow W \otimes W \quad ; \quad b_i \otimes b_j \mapsto
\begin{cases}
b_i \otimes b_i & i = j\\
0 & i \neq j
\end{cases}
\]
Recall that the graph Laplacian $L$ for a graph $G$ with adjacency matrix $A$ and its corresponding diagonal matrix $D$, defined as:
\[
L = D - A
\]
Then, using our generalized diagonalization procedure, we can compute generalized graph Laplacian $L(G)$ of our graph embedding $G$:
\[
L(G) = Diag_B(G) - G
\]
where $Diag_B$ is the diagonalization operator with respect to our orthogonal basis $B$. Given two orthogonal bases $B_1,B_2$ linked by an orthogonal change of basis matrix $P$, the induced change of basis in the tensor product (in coordinates) would be:
\[
e_i e_j^T \mapsto P e_i e_j^T P^T
\]
Hence, a diagonalized tensor in the new basis can be expressed as:
\[
Diag_{B_1}(\sum_{i,j} a_{ij} e_i e_j^T) = \sum_i a_{ii} e_i e_i^T \mapsto  \sum_i a_{ii} (P e_i) (P e_j)^T = P \sum_i (a_{ii} e_i e_i^T )P^T 
\]
Applying this to our generalized graph Laplacian $L(G)$ gives:
\[
L(G) = Diag_V(G) - G = PDP^T - PAP^T  = P(D - A)P^T = PLP^T
\]
Our generalized graph Laplacian is in fact the usual graph Laplacian, up to an orthogonal change of basis. It has the same spectrum as the usual Laplacian, and its eigenspaces are related by a distance-preserving transformation. We note that the graph Laplacian of a graph embedding might be expensive to compute, as we shall see in the next section when computing the diagonal term for tensor-spherical embeddings. Instead, we can perform the appropriate change of basis to convert our graph embedding into an adjacency matrix before computing its Laplacian, where diagonalization just requires us to copy the diagonal entries. 

If we drop the assumption of a exact orthonormal code and instead use a psuedo-orthonormal code, the only difference is finding an appropriate analogue to the diagonalization operator. Since we will have many more codes than dimensions, we cannot use the change-of-basis trick mentioned above and instead must resort to a more intensive computation. For any vertex code $v$, let $P_v = vv^T$ denote its associated projection matrix. Conjugation by $P_v$ has the following property:
\[
P_v (uw^T) P_v = \dprod{v}{u}\dprod{v}{w} vv^T = \begin{cases}
    vv^T & u=w=v\\
    \approx 0 & u\neq v \text{ or } w \neq v
\end{cases}
\]
Hence, for a pseudo-orthornomal code $V$ we have the following diagonalization operator:
\[
Diag_V(G) = \sum_{v \in V} P_v G P_v
\]
This involves multiple matrix conjugations and can be quite expensive. However, tensor-spherical embeddings are compressed relative to an adjacency matrix, and for sparse matrices computing the spectrum of the approximate Laplacian can be cheaper than computing the spectrum of the exact Laplacian. We will discuss this in the next subsection.

\subsection{Compressing Sparse Adjacency Matrices}
Our graph embeddings are generalizations of adjacency matrices in the idealized case of exact orthonormal codes. However, in practice we use pseudo-orthonormal codes to greatly reduce the dimension; hence, our method may also be viewed as a way to compress adjacency matrices while still retaining their functionality.

Usually, the dimension of the adjacency matrix grows as $d^2$ with the size of the vertex set $d$. For sparse matrices, where the number of edges $k$ is much lower than number of total possible connections $d^2$, this is very inefficient from a memory perspective since we are using $d^2$ parameters to represent $k << d^2$ edges. Instead, we may view adjacency matrices as special cases of our graph embeddings, representing a superposition of $k$ edges. In Section \ref{EdgeQuer} we showed our graph embeddings only need to scale with the number of edges $k$ in order to preserve accurate first-order graph operations. Hence, rather than using $d^2$ parameters to represent a $k$-sparse graph, tensor-spherical embeddings use the intuitively correct number of $k$ parameters.

The compression offered by tensor-spherical embeddings not only gives memory savings but can also decrease the computational cost of common graph operations. In the previous subsection we discussed an approximate graph Laplacian for tensor-spherical embeddings of the following form:
\[
L(G) = Diag_V(G) - G \qquad ; \qquad Diag_V(G) = \sum_{v \in V} P_v G P_v
\]
For a graph with $k$ edges, our graph embeddings need be roughly $\sqrt{k} \times \sqrt{k}$ matrices. The naive complexity of both matrix multiplication and eigenvalue decomposition for $n \times n$ matrices is $O(n^3)$ \cite{MatComplex}. Hence, computing the diagonal term in the Laplacian has complexity $O(d k^{3/2})$, and performing an eigenvalue decomposition has complexity $O(k^{3/2})$. In total, computing the graph Laplacian for a tensor-spherical embedding has complexity $O((d+1)k^{3/2})$. Since $k \propto d$ for sparse matrices, the complexity of the graph Laplacian is approximately $O(d^{5/2})$. In contrast, while the diagonal term is easy to compute for an adjacency matrix, its Laplacian will be a $d \times d$ matrix and finding its eigenvalues will have complexity $O(d^3)$. A similar analysis shows that edge composition and other multiplication procedures are cheaper for tensor-spherical embeddings than adjacency matrices.

\subsubsection{Comparison to Other Sparse Representations}
Tensor-spherical embeddings need to scale with the number of edges $k$ in order to preserve accurate graph representations. How does this compare to other compressed sparse matrix representations?

First, we introduce several common sparse matrix representations and their memory cost. One intuitive represention is Dictionary of Keys (DoK), where each non-zero entry has its row-column index as a key with its corresponding value. This method scales $O(k)$ since the number of key-value pairs equals the number of non-zero entries. The coordinate list representation is similar to DoK, except we now store 3-tuples $(row,column,value)$ for each non-zero entry. Finally, the Compressed Sparse Row (CSR) represents a sparse matrix using three arrays containing the non-zero values, their column indices, and the number of non-zero entries above each row respectively. The first two arrays are of length $k$, while the third is of length $d+1$ where we have $d$ total vertices. Since $d \propto k$ in a sparse matrix, we see that we also need $O(k)$ parameters. Tensor-spherical embeddings match the memory requirements of these sparse matrix representations.

Now, we analyze the complexity of matrix addition for both methods. For each of the sparse matrix representations discussed, matrix addition is $O(k)$ since we iterate through each value, check if they have matching indices, and sum their values if they do \cite{CormenAlgo} \cite{MatComplex}. Similarly, matrix addition of dense matrices of size $n$ is $O(n)$. Since tensor-spherical embeddings are of size $k$, their addition also has complexity $O(k)$.

Finally, we analyze the complexity of matrix multiplications for each representation. Consider two sparse $d \times d$ matrices $A$ and $B$ with $k_1$ and $k_2$ non-zero entries respectively. The worst-case scenario is when all the non-zero entries are concentrated on a single row and column. This gives an intuitive upper bound of $O(k_1 k_2)$, and this upper bound holds when $d^2 >> k_1,k_2$ \cite{SparseCompBorna}. As for tensor-spherical embeddings, the naive algorithm for the multiplication of two $d \times d$ matrices has complexity $O(d^3)$, while Strassen's algorithm improves this to $O(d^{2.372})$ \cite{CormenAlgo}. In order to retain accuracy for a second-order operation like edge composition and matrix multiplication, we saw in Section \ref{EdgeComp} that a $d \times d$ tensor-spherical embedding can store at most $d^{3/2}$ edges. Thus, in order to store $k$ edges we need them to be roughly $k^{2/3} \times k^{2/3}$ matrices, so their multiplication would be $O(k^2)$ assuming both embeddings have at most $k$ edges. This matches the complexity of the other representations when $k_1 \propto k_2$. 

In both memory cost and complexity of linear algebra operations, tensor-spherical embeddings match those of the considered sparse representations - DoK, coordinate list, and CSR format. However, tensor-spherical embeddings have at least one advantage over these methods. As vector representations, they enjoy all of the hardware and software optimizations developed for linear algebraic operations, including advances in parallel computing. Moreover, there has been extensive research in specialized hardware for hyperdimensional computing \cite{kleykoHDChard} \cite{HDCShearer} \cite{Li2016HyperdimensionalCW}, which can further augment the computational advantages of vector graph representations.

\subsection{Proximity to Adjacency Matrices}

In this section, we establish a rough bound on how close our graph embeddings are to adjacency matrices. Let us assume we $d$-dimensional vertex codes for $m$ vertices $\set{v_1,\cdots,v_m}$ where $m \leq d$; moreover, this vertex code is $\epsilon$-orthonormal, which holds with high probability for spherical codes from Theorem \ref{Optimality}. Then, we see that a graph embedding generated from such a code is close to some adjacency matrix.
\begin{theorem}
For a $d$-dimensional, $\epsilon$-orthonormal code $V = \set{v_1,\cdots,v_m}$ where $m \leq d$, let $G$ be any graph embedding using $V$ with $n$ distinct edges and $A$ be the corresponding adjacency matrix induced from this vertex ordering. Then, there exists an orthogonal matrix $P$ such that:
\[
||G-PAP^T||_F = ||P^T G P - A||_F < O(n m \epsilon)
\]
\begin{proof}
We use the Gram-Schmidt process to compute an orthogonal set of vectors $u_i$ from the vertex code $\set{v_i}$:
\[
u_i = v_i - \sum_{j=1}^{i-1} \frac{\dprod{v_i}{u_j}}{\dprod{u_j}{u_j}} u_j = v_i - \sum_{j=1}^{i-1} \dprod{v_i}{\overline{u}_j}\overline{u}_j
\]
where $\overline{u}_j$ is the unit-length version of $u_j$. By $\epsilon$-orthonormality, we know that $|\dprod{v_i}{v_j}| < \epsilon$ for all $i \neq j$. We know that $u_1 = \overline{u}_1 = v_1$. Then, for $u_2$: 
\[
1- \epsilon \leq ||u_2|| = ||v_2 - \dprod{v_2}{\overline{u}_1} \overline{u}_1|| \leq 1 + \epsilon
\]
and
\[
||v_2 - u_2|| = ||\dprod{v_2}{\overline{u}_1} \overline{u}_1|| < \epsilon
\]
Hence,
\[
||v_2 - \overline{u}_2|| \leq ||v_2 - u_2|| + ||u_2 - \overline{u}_2|| \leq 2\epsilon
\]
Similarly, for $||u_3|| = ||v_3 - \dprod{v_3}{\overline{u}_1} \overline{u}_1 - \dprod{v_3}{\overline{u}_2}\overline{u}_2||$, let us unravel the third term. We first analyze $\dprod{v_3}{u_2}$:
\[
|\dprod{v_3}{u_2}| = |\dprod{v_3}{v_2} - \dprod{v_2}{\overline{u}_1}\dprod{v_3}{\overline{u}_1}| < \epsilon + \epsilon^2
\]
This means that $|\dprod{v_3}{\overline{u}_2}| \leq \frac{\epsilon + \epsilon^2}{1-2\epsilon} = O(\epsilon)$.
This gives:
\[
1 - O(2\epsilon) = 1-\frac{O(2\epsilon)}{1-\epsilon} \leq ||u_3|| \leq 1 + O(2\epsilon)
\]
and so:
\[
||v_3 - \overline{u}_3|| \leq O(3 \epsilon)
\]
Repeating a similar analysis, we see that for $m$ $\epsilon$-orthogonal codes $\set{v_1,\cdots, v_m}$, there exists an orthonormal basis $\set{\overline{u}_1,\cdots,\overline{u}_m}$ such that $||v_i - \overline{u}_i|| < O(m \epsilon)$. Hence, let $G_V$ be any graph embedding using the vertex codes $\set{v_i}$ and $G_U$ be the corresponding matrix derived from swapping $v_i$ with $\overline{u}_i$. Letting $\delta_i = v_i - \overline{u}_i$, we compute the distance between the two matrices:
\begin{align*}
||G_V - G_U||_F &= ||\sum^n_{k=1} v_{i_k} v_{j_k}^T - \overline{u}_{i_k} \overline{u}_{j_k}^T||_F\\
&\leq \sum^n_{k=1} ||v_{i_k} v_{j_k}^T - \overline{u}_{i_k} \overline{u}_{j_k}^T || \\
&= \sum^n_{k=1} ||v_{i_k} v_{j_k}^T - (v_{i_k}- \delta_{i_k}) (v_{j_k} - \delta_{j_k})^T|| \\
&= \sum^n_{k=1} ||v_{i_k} v_{j_k}^T -[v_{i_k} v_{j_k}^T + \delta_{i_k}\delta_{j_k}^T - v_{i_k}\delta_{j_k}^T - \delta_{i_k}v_{j_k}^T] ||\\
&\leq \sum^n_{k=1} ||\delta_{i_k}\delta_{j_k}^T|| + ||v_{i_k}\delta_{j_k}^T|| + ||\delta_{i_k}v_{j_k}^T||\\
&\leq O(n m \epsilon)
\end{align*}
As $G_U$ is a sum of the outer products between the orthonormal $\overline{u}_i$, we see that $U^T G_U U$ is an adjacency matrix. Since the Frobenius norm is preserved under orthogonal basis change, we see that $G_V$ is $O(n m \epsilon)$ close to an adjacency matrix with respect to the Frobenius norm.
\end{proof}
\end{theorem}

\section{Experiments}\label{Experiments}
The interested reader is encouraged to look at the accompanying code if they wish to replicate the experiments of this section. We perform several experiments confirming several theoretical results. Firstly, we validate our claim that the dimension of tensor-spherical embeddings needs to scale with the number of edges $k$ rather than the size of the vertex set $n$ to retain accurate graph operations, even in the presence of vertex overlap as in Sections \ref{EQ:ConnCorrection} and \ref{EC:ConnCorr}. This supports our earlier claim that tensor-spherical embeddings offer a drastic compression of adjacency matrices, especially for sparse matrices. Secondly, we compare the accuracy of the query and edge between the tensor-spherical and Hadamard-Rademacher scheme when increasing the number of edges in the embeddings. 

\subsection{Effects of Vertex Set Size and Vertex Connectivity}
We look at the effect of the number of vertices and vertex connectivity on the accuracy of a edge query and edge composition for tensor-spherical embeddings. Specifically, for a given codebook size we generate a graph by sampling vertex pairs, without replacement, from the codebook. This procedure will cause some of the nuisance edges share a vertex with the query edge: for a graph with $k$ edges generated from a codebook of size $n$, on average approximately $\frac{2k}{n}$ of the edges will share one vertex with the query edge. For both the edge query and edge composition, we looked at the performance in both a positive query - where query edge was present in the graph - and a spurious query - where the query edge was absent in the graph.

Here are the details of the experiment. We fixed the number of edges at 64 and varied the vertex codebook size from 64 to 2048. The vertex code dimension was fixed to 16, meaning our graph embeddings were $16 \times 16$ matrices. Each codebook was generated by independently sampling uniformly from the 64-dimensional unit hypersphere, each graph was generated by randomly selecting 64 pairs of vertices from the codebook without replacement. For the edge query, the query edge was subtracted in the spurious case. For the edge composition, in the positive case an extra edge was added whose source vertex was the query edge's target vertex and target vertex was separately generated; this was done to prevent a repeat composable edge. At each codebook size, we repeated the edge query/edge composition 200 times, recording both the average query score and its standard deviation. The experiment results are shown in Table \ref{tab:CB_VerConn}.

The results show that the size of the vertex set has a positive impact on the accuracy of tensor-spherical embeddings, as the query scores' standard deviation decreases as the number of vertices increases. This is driven by decreasing vertex overlap with the query edge as the vertex set increases. This confirms that in order to retain accuracy, tensor-spherical embeddings do not scale with the size of the vertex set. Interestingly, our graph embeddings were $16 \times 16$ matrices yet were able to accurately store 64 edges in superposition. A corresponding adjacency matrix would need to be a $128 \times 128$ matrix - an 64-fold increase in the number of parameters used.

\begin{table}[!htb]
    \centering
    \begin{tabular}{|c|c|c|c|c|c|c|c|}
    \hline
    \multicolumn{8}{|c|}{\textbf{Effect of Vertex Set Size and Vertex  Connectivity}} \\
    \hline
    \hline
    \multicolumn{8}{|c|}{\textbf{Edge Query}} \\
    \hline
    Vertex Set Size & 32 & 64 & 128 & 264 & 512 & 1024 & 2048\\
    \hline
    Positive Query & .99 & 1.04 & .97 & .99& 1.03 & 1.03 &1.02\\
    & (.60) & (.60) & (.54) & (.51) & (.53) & (.50) & (.51)\\
    \hline
    Spurious Query & -.11 & .06 & -.02 & -.01 & .03 & -.03 & -.06\\
    & (.64) & (.60) & (.54) & (.49) & (.51) & (.51) & (.52)\\
    \hline
    \multicolumn{8}{|c|}{} \\
    \hline
    \multicolumn{8}{|c|}{\textbf{Edge Composition}} \\
    \hline
    % Vertex Set Size & 32 & 64 & 128 & 264 & 512 & 1024 & 2048\\
    % \hline
    Positive Query & 1.03 & 1.26 & 1.13 & 1.06 & 1.09 & .78 & 1.05\\
    & (1.14) & (1.68) & (1.30) & (1.32) & (1.28) & (1.17) & (1.21)\\
    \hline
    Spurious Query & .03 & .10 & .05 & -.03 & .04 & .04 & .05\\
    & (.92) & (1.40) & (.99) & (1.06) & (1.03) & (1.03) & (.95)\\
    \hline
    \end{tabular}
    
    % \bigskip\bigskip
    % \begin{tabular}{|c|c|c|c|c|c|c|c|}
    % \hline
    % \multicolumn{8}{|c|}{Edge Composition: Vertex Set Size and Vertex Connectivity} \\
    % \hline
    % Vertex Set Size & 32 & 64 & 128 & 264 & 512 & 1024 & 2048\\
    % \hline
    % Positive Query & .79(1.56) & 1.12(1.26) & 1.05(1.33) & 1.02(1.21) &..90(1.01) &1.04(1.26) &1.00(1.30)\\
    % \hline
    % Spurious Query & -.21(1.31) & -.07(1.06) & -.06(1.06) & .12(1.00) & -.06(.99) & .04(.92) & -.03(1.01)\\
    % \hline
    % \end{tabular}

    \caption{Each cell contains the average of 200 trials, with the standard deviation in parentheses. The positive query has an ideal score of 1, while the spurious query has an ideal score of 0. For codebook size $n$, on average approximately $\frac{128}{n}$ of the edges in the generated graph share a vertex with the query edge.}
    \label{tab:CB_VerConn}
\end{table}

\clearpage
\subsection{Effect of Edgeset Size on Accuracy}
We also ran simulations to confirm the memory-capacity results of Sections \ref{EdgeQuer} and \ref{EdgeComp}. For both tensor-spherical embeddings and Hadamard-Rademacher embeddings, we looked their performance in the edge query and edge composition task. The experimental procedure was the same as that of the last section with some key differences. Firstly, we now fix the vertex codebook size to 64 and instead vary the number of edges in superposition from 8 to 512. Secondly, while we fixed the vertex code dimension to 16 for tensor-spherical embeddings, for Hadamard-Rademacher embeddings we set the vertex code dimension to $16^2 =256$. This was done so each embedding used the same number of parameters. Finally, note that from Sections \ref{EdgeQuer} and \ref{EdgeQuer} the query scores of Hadmard-Rademacher embeddings are scaled by their dimension. Hence, we divided the Hadamard-Rademacher query scores by 256 to normalize them and make them comparable to the tensor-spherical scores. The results are shown in Table \ref{tab:MemCapComparison}.

The results coincide with our theoretical analyses, and in particular they demonstrate the differing efficiency of the tensor and Hadamard products in edge composition.  From the results of Section \ref{EdgeComp}, we showed that Hadmard-Rademacher embeddings of size $n$ can store at most $O(\sqrt{n})$ edges; for embeddings of size 256 in this experiment, the number of edges they can store is proportional to 16. We see that for embeddings that contain 8 edges, the positive query scores already significantly deviate from the ideal value of 1, and the standard deviation magnitudes dwarf the average when the embeddings contain at least 48 edges. Conversely, tensor-spherical of size $n$ can store $O(n^{\frac{3}{4}})$ edges; for embeddings of size 256 in this experiment, the number of edges they can accurately store is proportional to 64. Correspondingly, tensor-spherical edge composition scores do not significantly deviate from the ideal value of 1 until 264 edges, and the error standard deviations do not scale as poorly as those of Hadamard-Rademacher embeddings.

\begin{table}[!htb]
    \centering
    \begin{tabular}{|c|c|c|c|c|c|c|c|c|}
    \hline
    \multicolumn{9}{|c|}{\textbf{Effect of Edgeset Size}} \\
    \hline
    \hline
    \multicolumn{9}{|c|}{\textbf{Tensor-Spherical: Edge Query}} \\
    \hline
    Edgeset Size & 8 & 16 & 32 & 48 & 64 & 128 & 264 & 512\\
    \hline
    Positive Query & 1.02 & .98  & 1.05  & 1.00  & 1.03 & 1.01  & 1.14 & 1.06 \\
    &  (.22) & (.30) & (.52) & (.48) & (.56) & (.78) &  (1.21) & (1.78)\\
    \hline
    Spurious Query & .01  & .01  & .04  & .01  & .01  & .08  & .01  & -.13 \\
    & (.22) & (.30) & (.38) & (.55) & (.59) & (.86) &  (1.26) & (1.81)\\
    \hline
    \multicolumn{9}{|c|}{} \\
    \hline
    \multicolumn{9}{|c|}{\textbf{Hadamard-Rademacher: Edge Query}} \\
    \hline
    Positive Query & 1.02  & .99  & 1.03  & 1.01  & .98  & 1.05  & 1.12  & 1.35 \\
    & (.17) & (.29) & (.32) & (.42) & (.52) & (.70) & (1.30) & (1.98)\\
    \hline
    Spurious Query & .00  & .01  & .04  & -.03  & .03  & .02  & .15  & .36 \\
    & (.18) & (.26) & (.36) & (.42) & (.54) & (.81) & (1.15) & (1.83)\\
    \hline
    \multicolumn{9}{|c|}{} \\
    \hline
    \hline
    \multicolumn{9}{|c|}{} \\
    \hline
    \multicolumn{9}{|c|}{\textbf{Tensor-Spherical: Edge Composition}} \\
    \hline
    Positive Query & 1.05  & 1.10  & 1.11  & 1.09  & 1.20 & 1.04  & 2.53  & 7..92 \\
    & (.36) & (.54) & (.91) & (1.13) & (1.53) &  (3.12) & (6.59) & (14.9)\\
    \hline
    Spurious Query & .01  & .00  & .07  & -.02  & .11  & .32  & .54  & 3.87 \\
    & (.13) & (.29) & (.59) & (.97) & (1.34) & (2.61) & (5.94) & (14.36)\\
    \hline
    \multicolumn{9}{|c|}{} \\
    \hline
    \multicolumn{9}{|c|}{\textbf{Hadamard-Rademacher: Edge Composition}} \\
    \hline
    Positive Query & 2.23  & 2.14  & 3.05  & 3.01  & 3.80  & 7.91  & 11.34  & 45.95 \\
    & (.89) & (1.85) &  (3.65) & (5.15) & (6.75) & (13.39) & (32.18) & (79.99)\\
    \hline
    Spurious Query & .09  & .38  & .52  & .76  & 1.77  & 2.12  & 11.46  & 43.12 \\
    & (.71) & (1.41) & (.34) & (5.06) & (6.97) & (15.16) & (33.62) & (74.72)\\
    \hline
    \end{tabular}
    
    % \bigskip\bigskip
    % \begin{tabular}{|c|c|c|c|c|c|c|c|}
    % \hline
    % \multicolumn{8}{|c|}{Edge Composition: Vertex Set Size and Vertex Connectivity} \\
    % \hline
    % Vertex Set Size & 32 & 64 & 128 & 264 & 512 & 1024 & 2048\\
    % \hline
    % Positive Query & .79(1.56) & 1.12(1.26) & 1.05(1.33) & 1.02(1.21) &..90(1.01) &1.04(1.26) &1.00(1.30)\\
    % \hline
    % Spurious Query & -.21(1.31) & -.07(1.06) & -.06(1.06) & .12(1.00) & -.06(.99) & .04(.92) & -.03(1.01)\\
    % \hline
    % \end{tabular}

    \caption{Edge query and edge composition scores for tensor-spherical embeddings and Hadamard-Rademacher embeddings as the number of stored edges is increased. The average score of 200 trials is recorded, with the standard deviation shown in parentheses. Both embeddings had dimension 256. For both graph operations, the positive query has an ideal score of 1 and the spurious query has an ideal score of 0.}
    \label{tab:MemCapComparison}
\end{table}

\clearpage
\section{Conclusion}
In this paper, we showcased the many nice theoretical properties of tensor-spherical graph embeddings. As compressed generalizations of adjacency matrices, they are capable of a wide range of graph operations while sidestepping the scaling issues that plague adjacency matrices. In particular, we demonstrated their utility in representing sparse matrices, showing they match other sparse matrix representations (dictionary of keys, coordinate list, compressed sparse row) in both parameter usage and complexity of graph operations. Our theoretical and experimental results confirmed that tensor-spherical embeddings can drastically compress adjacency matrices while still retaining accurate graph functionality. The results contained in this paper are incomplete, and the utility of tensor-spherical embeddings as sparse matrix representations is an interesting direction of inquiry.

From the hyperdimensional computing view, we also gave several theoretical properties of the tensor product that make it an attractive binding operation for graph embeddings. First, we established a mathematical link between the tensor product and the superposition principle, showing that it is the most general and hence most expressive binding operation among all binding operations that respect superposition. Our analysis also showed that tensor-spherical embeddings have the same memory-capacity scaling and superior memory-capacity ratio as Hadamard-Rademacher embeddings, suggesting the Hadamard product and its related binding operations do not offer any actual memory savings. We demonstrated that this is due in part to the tensor product's ability to leverage pseudo-orthogonality to control error terms, which the Hadamard product lacks. 

\section{Acknowledgements}
We look like to thank Bruno Olshausen, Giles Hooker, and Denis Kleyko for their help and many fruitful discussions. Their suggestions helped clarify the results of this work and contextualize it, as well as raise interesting directions of inquiry. We would also like to thank the Redwood Center for Theoretical Neuroscience, whose members' research and talks inspired this work.
\clearpage

\bibliography{refs.bib}

\begin{thebibliography}{10}

\bibitem{SparseCompBorna}
Keivan Borna and Sohrab Fard.
\newblock A note on the multiplication of sparse matrices.
\newblock {\em Open Computer Science}, 4(1):1--11, 2014.

\bibitem{Boucheron2004}
St{\'e}phane Boucheron, G{\'a}bor Lugosi, and Olivier Bousquet.
\newblock {\em Concentration Inequalities}, pages 208--240.
\newblock Springer Berlin Heidelberg, Berlin, Heidelberg, 2004.

\bibitem{Cheung2019SuperpositionOM}
Brian Cheung, Alex Terekhov, Yubei Chen, Pulkit Agrawal, and Bruno~A.
  Olshausen.
\newblock Superposition of many models into one.
\newblock {\em ArXiv}, abs/1902.05522, 2019.

\bibitem{clarkson2023capacity}
Kenneth~L. Clarkson, Shashanka Ubaru, and Elizabeth Yang.
\newblock Capacity analysis of vector symbolic architectures, 2023.

\bibitem{CormenAlgo}
Thomas~H. Cormen, Charles~E. Leiserson, Ronald~L. Rivest, and Clifford Stein.
\newblock {\em Introduction to Algorithms, Third Edition}.
\newblock The MIT Press, 3rd edition, 2009.

\bibitem{Gayler2009ADB}
Ross~W. Gayler and Simon~D. Levy.
\newblock A distributed basis for analogical mapping.
\newblock 2009.

\bibitem{JLLemma}
William~B. Johnson and Joram Lindenstrauss.
\newblock Extensions of lipschitz mappings into a hilbert space.
\newblock {\em Contemporary Mathematics}, 26, 1984.

\bibitem{kanerva_sdm}
Pentti Kanerva.
\newblock {\em Sparse Distributed Memory}.
\newblock MIT Press, Cambridge, MA, USA, 1988.

\bibitem{Kang2022RelHDAG}
Jaeyoung Kang, Minxuan Zhou, Abhinav Bhansali, Weihong Xu, Anthony Thomas, and
  Tajana Rosing.
\newblock Relhd: A graph-based learning on fefet with hyperdimensional
  computing.
\newblock {\em 2022 IEEE 40th International Conference on Computer Design
  (ICCD)}, pages 553--560, 2022.

\bibitem{CauchyRatio}
John Kermond.
\newblock {\em An Introduction to the Algebra of Random Variables}, pages
  1--16.
\newblock The Mathematical Association of Victoria, 2010.

\bibitem{HDCShearer}
Behnam Khaleghi, Sahand Salamat, Anthony Thomas, Fatemeh Asgarinejad, Yeseong
  Kim, and Tajana Rosing.
\newblock Shearer: Highly-efficient hyperdimensional computing by
  software-hardware enabled multifold approximation.
\newblock In {\em Proceedings of the ACM/IEEE International Symposium on Low
  Power Electronics and Design}, ISLPED '20, page 241–246, New York, NY, USA,
  2020. Association for Computing Machinery.

\bibitem{kleykoHDChard}
Denis Kleyko, Mike Davies, E.~Paxon Frady, Pentti Kanerva, Spencer~J. Kent,
  Bruno~A. Olshausen, Evgeny Osipov, Jan~M. Rabaey, Dmitri~A. Rachkovskij,
  Abbas Rahimi, and Friedrich~T. Sommer.
\newblock Vector symbolic architectures as a computing framework for nanoscale
  hardware.
\newblock In {\em Proceedings of the IEEE}, volume~10, page 1538–1571, 2022.

\bibitem{kleyko_VSA}
Denis Kleyko, Dmitri~A. Rachkovskij, Evgeny Osipov, and Abbas~Jawdat Rahim.
\newblock A survey on hyperdimensional computing aka vector symbolic
  architectures, part ii: Applications, cognitive models, and challenges.
\newblock {\em ArXiv}, abs/2112.15424, 2021.

\bibitem{lang02}
Serge Lang.
\newblock {\em Algebra}.
\newblock Springer, New York, NY, 2002.

\bibitem{JLopt}
Kasper~Green Larsen and Jelani Nelson.
\newblock Optimality of the johnson-lindenstrauss lemma.
\newblock In {\em 2017 IEEE 58th Annual Symposium on Foundations of Computer
  Science (FOCS)}, pages 633--638, 2017.

\bibitem{Li2016HyperdimensionalCW}
Haitong Li, Tony~F. Wu, Abbas Rahimi, Kai-Shin Li, Miles Rusch, Chang-Hsien
  Lin, Juo-Luen Hsu, Mohamed~M. Sabry, Sukru~Burc Eryilmaz, Joon Sohn,
  Wen-Cheng Chiu, Min-Cheng Chen, Tsung-Ta Wu, Jia-Min Shieh, W.~K. Yeh, Jan~M.
  Rabaey, Subhasish Mitra, and H.~S.~Philip Wong.
\newblock Hyperdimensional computing with 3d vrram in-memory kernels:
  Device-architecture co-design for energy-efficient, error-resilient language
  recognition.
\newblock {\em 2016 IEEE International Electron Devices Meeting (IEDM)}, pages
  16.1.1--16.1.4, 2016.

\bibitem{MatComplex}
Yan Li, Sheng-Long Hu, Jie Wang, and Zheng-Hai Huang.
\newblock An introduction to the computational complexity of matrix
  multiplication.
\newblock {\em Journal of the Operations Research Society of China},
  8(1):29--43, 2020.

\bibitem{Ma2018HolisticRF}
Yunpu Ma, Marcel Hildebrandt, Volker Tresp, and Stephan Baier.
\newblock Holistic representations for memorization and inference.
\newblock In {\em UAI}, 2018.

\bibitem{MarchalBetaSub}
Olivier Marchal and Julyan Arbel.
\newblock On the sub-gaussianity of the beta and dirichlet distributions.
\newblock {\em Electronic Communications in Probability}, 22(none), jan 2017.

\bibitem{UniformRatio}
George Marsaglia.
\newblock Ratios of normal variables and ratios of sums of uniform variables.
\newblock {\em Journal of the American Statistical Association},
  60(309):193--204, 1965.

\bibitem{NickelTensor}
Maximilian Nickel, Xueyan Jiang, and Volker Tresp.
\newblock Reducing the rank in relational factorization models by including
  observable patterns.
\newblock In Z.~Ghahramani, M.~Welling, C.~Cortes, N.~Lawrence, and K.Q.
  Weinberger, editors, {\em Advances in Neural Information Processing Systems},
  volume~27. Curran Associates, Inc., 2014.

\bibitem{Nickel2016HolographicEO}
Maximilian Nickel, Lorenzo Rosasco, and Tomaso~A. Poggio.
\newblock Holographic embeddings of knowledge graphs.
\newblock In {\em AAAI}, 2016.

\bibitem{nickel2013logistic}
Maximilian Nickel and Volker Tresp.
\newblock Logistic tensor factorization for multi-relational data, 2013.

\bibitem{Nunes2022GraphHDEG}
Igor~O. Nunes, Mike Heddes, Tony Givargis, Alexandru Nicolau, and Alexander~V.
  Veidenbaum.
\newblock Graphhd: Efficient graph classification using hyperdimensional
  computing.
\newblock {\em 2022 Design, Automation \& Test in Europe Conference \&
  Exhibition (DATE)}, pages 1485--1490, 2022.

\bibitem{PlateHRR}
T.A. Plate.
\newblock Holographic reduced representations.
\newblock {\em IEEE Transactions on Neural Networks}, 6(3):623--641, 1995.

\bibitem{Poduval_graph_embed}
Prathyush Poduval, Haleh Alimohamadi, Ali Zakeri, Farhad Imani, M.~Hassan
  Najafi, Tony Givargis, and Mohsen Imani.
\newblock Graphd: Graph-based hyperdimensional memorization for brain-like
  cognitive learning.
\newblock {\em Frontiers in Neuroscience}, 16, 2022.

\bibitem{Schlegel_2021}
Kenny Schlegel, Peer Neubert, and Peter Protzel.
\newblock A comparison of vector symbolic architectures.
\newblock {\em Artificial Intelligence Review}, 55(6):4523--4555, dec 2021.

\bibitem{Smolensky_tensor}
Paul Smolensky.
\newblock Tensor product variable binding and the representation of symbolic
  structures in connectionist systems.
\newblock {\em Artificial Intelligence}, 46(1):159--216, 1990.

\bibitem{Thomas2020ATP}
Anthony Thomas, Sanjoy Dasgupta, and Tajana Simunic.
\newblock A theoretical perspective on hyperdimensional computing.
\newblock {\em J. Artif. Intell. Res.}, 72:215--249, 2020.

\bibitem{ZulaikaTensor}
Unai Zulaika, Aitor Almeida, and Diego L\'{o}pez-de Ipi\~{n}a.
\newblock Regularized online tensor factorization for sparse knowledge graph
  embeddings.
\newblock {\em Neural Comput. Appl.}, 35(1):787–797, sep 2022.

\end{thebibliography}
\bibliographystyle{plain}

\end{document}